\documentclass[21,reqno]{amsart}

\usepackage[english, activeacute]{babel}
\usepackage{amsmath,amsthm,amssymb}
\usepackage{epsfig}
\usepackage{amssymb}
\usepackage{latexsym}
\usepackage{amsfonts}
\usepackage[active]{srcltx}
\usepackage{pstricks,pst-grad}

\title{Characterization of the Anderson metal-insulator transition for non ergodic
operators and application}
\author{Constanza Rojas-Molina}
\address{Universit\'e de Cergy-Pontoise, UMR CNRS 8088, F-95000 Cergy-Pontoise, France}
\email{constanza.rojas-molina@u-cergy.fr}
 \thanks{The author was partially supported by ANR BLAN 0261.  The author would like to thank F. Germinet for his valuable support, interesting discussions and helpful remarks.}

\keywords{random Schr\"odinger operators, Anderson transition, Delone operators}

\numberwithin{equation}{section}
\newtheorem{thm}{Theorem}[section]

\newtheorem{lem}[thm]{Lemma}

\theoremstyle{definition}
\newtheorem{defn}{Definition}

\theoremstyle{remark}
\newtheorem{rem}{Remark}
[section]

\newcommand{\be}{\begin{equation}}
\newcommand{\ee}{\end{equation}}
\newcommand{\ba}{\begin{array}}
\newcommand{\ea}{\end{array}}
\newcommand{\bal}{\begin{align}}
\newcommand{\eal}{\end{align}}
\newcommand{\bea}{\begin{eqnarray}}
\newcommand{\eea}{\end{eqnarray}}
\newcommand{\bee}{\begin{eqnarray*}}
\newcommand{\eee}{\end{eqnarray*}}
\newcommand{\tr}{{\rm{tr}}}
\newcommand{\dist}{\mbox{dist}}
\renewcommand{\P}{\mathbb P}

\newcommand{\Rd}{\mathbb R^d}
\newcommand{\R}{\mathbb R}
\newcommand{\C}{\mathbb C}
\renewcommand{\H}{\mathcal H}
\newcommand{\X}{\mathcal X}
\newcommand{\Homega}{H_\omega}

\newcommand{\Prob}[1]{\mathbb P\left(#1\right)}
\newcommand{\norm}[1]{\Vert #1 \Vert}
\newcommand{\abs}[1]{\left| #1 \right|}

\newcommand{\Lp}[1]{\textrm{L}^2(#1)}
\renewcommand{\L}{\Lambda}
\newcommand{\set}[1]{\left\{ #1 \right\}}
\newcommand{\angles}[1]{\langle #1 \rangle}

\begin{document}

\maketitle \markboth {Anderson metal-insulator transition for non ergodic
operators}{Anderson metal-insulator transition for non ergodic operators}

\renewcommand{\sectionmark}[1]{}
\begin{abstract}
We study the Anderson metal-insulator transition for non ergodic
random Schrödinger operators in both annealed and quenched regimes, based on a
dynamical approach of localization, improving known results for
ergodic operators into this more general setting.  In the procedure, we
reformulate the Bootstrap Multiscale Analysis of Germinet and Klein to
fit the non ergodic setting.  We obtain uniform Wegner Estimates needed to perform this adapted Multiscale Analysis in the case of Delone-Anderson type potentials, that is, Anderson potentials modeling aperiodic solids, where the impurities lie on a Delone set rather than a lattice, yielding a break of ergodicity.  As an application we study the Landau operator with a Delone-Anderson potential and show the existence of a mobility edge between regions of dynamical localization and dynamical delocalization.

\end{abstract}
\tableofcontents

\section{Introduction}

Under the effect of a random pertubation, the spectrum of an ergodic Schr\"odinger operator is expected to undergo a transition where we can identify two distinct regimes: the insulator region, characterized by localized states and the metallic region, characterized by extended states.  The passage from one to the other under a certain disorder regime is known as the Anderson metal-insulator transition.  Although a precise spectral description of this phenomena is still out of reach, this transition is better characterized in terms of its dynamical properties.  Germinet and Klein tackled this problem in \cite{GK3} by introducing a local transport exponent $\beta(E)$ to measure the spreading of a wave packet initially localized in space and in energy evolving under the effect of the random operator. This provides a proper dynamical characterization of the metal-insulator transition, and the mobility edge, i.e. the energy where the transition occurs, is shown to be a discontinuity point of $\beta(E)$.  

Since ergodicity  is a basic feature in the theory of random Schr\"odinger operators, Germinet and Klein's work was done in that framework.  However, more real models may lack this fundamental property, examples of this kind of systems are Schr\"odinger operators with Anderson-type potentials where the random variables are not i.i.d. or where impurities are located in aperiodic discrete sets.  The first case (sparse models, decaying randomness, surfacic potentials) has been studied in \cite{BKS}, \cite{BdMSS}, \cite{KV}, \cite{S}, while the second case (Delone-Anderson type potentials) has been treated in \cite{BdMNSS}.  In the deterministic case, Delone operators have been studied with a dynamical systems appproach in \cite{KLS} and \cite{MR}.

We aim to study the Anderson metal-insulator characterization in a general non ergodic setting, with minimal requirements on the model to fit the dynamical characterization of localization/delocalization using the local transport exponent $\beta(E)$, extending the results of \cite{GK3} to the non ergodic models mentioned above.  The main tool in the study of the transport transition is the Multiscale Analysis (MSA), initially developped by Fr\"olich and Spencer \cite{FrS}, it has been improved over the last three decades to its strongest version so far, the Bootstrap MSA by Germinet and Klein  \cite{GK1}.  The Bootstrap MSA yields among other features strong dynamical localization in the Hilbert-Schmidt norm, and so it can be used to characterize the set of energies where the transport exponent is zero, that is associated to dynamically localized states \cite{GK3}, but since it was originally developped in the frame of ergodic operators it is not suitable when there is lack of ergodicity, so we adapt it to our model.  What completes the dynamical characterization is the fact that, in the ergodic case, slow transport in average over the randomness, the so-called annealed regime, implies dynamical localization.  This holds in our new setting and, moreover, this can be improved and it can be shown that it is enough to have slow transport with a good probability, that is, in a quenched regime, to obtain dynamical localization, so in both quenched and annealed regimes the metal-insulator transition can be characterized in an analog way.   There are examples related to the Parabolic Anderson model where the behavior of the solution in both regimes differ from each other and this can depend on the density of the random variables \cite{GaKo}.

We obtain uniform Wegner estimates nedeed for the adapted version of the Bootstrap MSA for both the Laplacian and the Landau operator with Delone-Anderson potentials, that is, Anderson potentials where the impurities are placed in an \emph{a priori} aperiodic set, called a \emph{Delone} set.  It is known that a way to obtain Wegner estimate is to ``lift'' the spectrum by considering the random Hamiltonian as a negative perturbation of a periodic Hamiltonian whose spectrum starts above a certain energy above the bottom of the spectrum of the original free Hamiltonian (called fluctuation boundary).  In this way the Wegner estimate is obtained ``outside`` the spectrum of the periodic operator, as in \cite{BdMLS}.  We stress the fact that this approach is not convenient in our case since we have no information on where the fluctuation boundary lies.  On the other hand, \cite{CHK} and \cite{CHK2} take a different approach by using a unique continuation property to prove Wegner estimates without a covering condition on the single-site potential, and not using fluctuation boundaries.   The results in \cite{CHK} rely strongly on the periodicity of the lattice and the use of Floquet theory, which, again, cannot be used in our model since our set of impurities is aperiodic. However, this was improved in \cite{CHKR} to obtain a positivity estimate for the Landau Hamiltonian that does not rely on Floquet theory, which makes it convenient for our setting.  In the case of the free Laplacian (see \cite{G}) we use a \emph{spatial averaging} method as in \cite{GKH}, \cite{BoK} to prove the required positivity estimate, thus bypassing the use of Floquet theory.  As a result we obtain a uniform Wegner estimate at the bottom of the spectrum in an interval whose lenght depends only on the Delone set parameters and not in the disorder parameter $\lambda$.  We also obtain Wegner estimates in the case where the background hamiltonian is either periodic or the Landau operator.  For the latter, and as an application of the main results, we can show the existence of a metal-insulator transition, as expected from the ergodic case \cite{GKS}. Since the lattice is a particular case of a Delone set, these results imply in particular those of the ergodic setting.  By the lack of ergodicity we cannot make use of the Integrated Density of States to prove the existence of a non random spectrum for $H_\omega$, nor use the characterization of the spectrum in terms of the spectra of periodic operators as done in \cite{GKS2} to locate the spectrum in the Landau band.  Therefore, to show our results are not empty we need to prove that we can almost surely find spectrum near the band edges, which is done adapting an argument in \cite[Appendix B]{CH} in a not necesarily perturbative regime of the disorder parameter $\lambda$.  We stress that we consider a general Delone set and do not assume any geometric property, like repetitivity or finite local complexity.  These features, however, might be needed for further results, for example, related to the Integrated Density of States (see \cite{MR}, \cite{LS3}, \cite{LV}).

The present note is organized as follows: in Section 2 we adapt the Bootstrap MSA to fit our new  setting.  In Section 3 we prove the results on the dynamics in both annealed and quenched regimes.  In Section 4 we prove uniform Wegner estimates for Delone-Anderson random Schr\"odinger operators.  In Section 5, in the lines of \cite{GKS} we proof the existence of a metal-insulator transition for  a Landau Hamiltonian with a Delone-Anderson potential and the existence of alsmot sure spectrum near the band edges, that has non empty intersection with the localization region.

\section{Main results}

%

For $x\in\Rd$ we denote by $\norm{x}$ the usual euclidean norm 
while the supremum norm is defined as $\abs{x}_\infty=\displaystyle\max_{1\leq i\leq
d}\abs{x_i}$,
where $\abs{ \cdot }$ stands for absolute value.

Given $x\in\Rd$ and $L>0$ we denote by $B(x,L)$ the ball of
center $x$ and radius $L$ in the $\norm{\cdot}$-norm, while the
set

\[ \L_L(x)=\left\{ y\in\Rd : |y-x|_{\infty}<\frac{L}{2}\right\} \]
defines the cube of side $L$ centered at $x$, also denoted as $\L_{x,L}$ .  We denote the
volume of a Borel set $\L\subset\Rd$ with respect to
the Lebesgue measure as $|\L|=\int_{\Rd}\chi_\L(x)d^dx $, where $\chi_\L$ is the
characteristic function of the set $\L$.  We will often write
$\chi_{x,L}$ for $\chi_{\L_L(x)}$ and denote by $\norm{f}_{x,L}$ or $\norm{f}_{\L_L(x)}$  the norm of $f$ in $\Lp{\L_{x,L}}$.

\bigskip
 We denote by $\mathcal C^\infty_c(\L)$ the vector space of
real-valued infinitely differentiable functions with compact support
contained in $\L$, with $ C^\infty_{c,+}(\L) $ being the
subclass of nonnegative functions.

\bigskip

We denote by $\mathcal B(\H)$ the Banach space of bounded
linear operators on the Hilbert space $\H$. For a closed,
densely defined operator A with adjoint $A^*$, we denote its domain
by $\mathcal D(A)\subset \Lp{\L}$ and by $\norm{A}=\sup\{\norm{A\phi};\quad \norm{\phi}_2=1\}$ its
(uniform) norm if bounded. We define its absolute value by
$\abs{A}=\sqrt{A^*A}$ and, for $p>1$, we define its (Schatten) $p$-norm
in the Banach space $\mathcal J_p(\Lp{\L})$ as $\norm{A}_p=(\tr |A|^p)^{1/p}$. In particular, $\mathcal J_1$ is the
space of trace-class operators and $\mathcal J_2$, the space of
Hilbert-Schmidt operators.  We write $\left<x\right>=\sqrt{(1+\norm{x}^2)}$
and use $\left<X\right>$
 to denote the operator given by multiplication by the function
$\left<x\right>$.

For convenience we denote a constant $C$ depending only on the
parameters a,b,... by $C_{a,b,...}$. \\

\bigskip

We consider a random Sch\"odinger operator of the form

\be\label{ranop}
 H_{\omega}= H_0 +\lambda V_\omega \quad \textrm{on}\
 \Lp{\R},\ee
where $H_0$ is the free Hamiltonian, $\lambda$ measures the disorder strength
which in the
following we consider fix, and $V_\omega$, called \emph{random
potential}, is the operator multiplication by $V_\omega$,
such that $\{V_\omega(x): x\in\Rd\}$ is a real-valued measurable process
on a complete probability space $(\Omega, \mathcal F, \mathbb P)$ having the
following properties:

\begin{itemize}
 \item[(R)] $V_\omega=V_\omega^+ +V_\omega^-$, where $V_\omega^+$ and
$V_\omega^-$ are real valued measurable processes on $\Omega$ such that for
$\mathbb P-$a.e. $\omega$ : $0\leq V_\omega^+\in \mbox{L}_{loc}^1(\Rd)$ and
$V_\omega^-$ is relatively form-bounded with respect to $-\Delta$, with relative
bound $<1$, i.e. there are nonegative constants $\Theta_1<1$ and $\Theta_2$
independent of $\omega$ such that for all $\psi\in\mathcal D(\nabla)$ we have

\[ |\left< \psi,V_\omega^-\psi \right>|\leq
\Theta_1\Vert\nabla\psi\Vert^2+\Theta_2\Vert\psi\Vert^2 \mbox{ for $\mathbb
P$-a.e. $\omega$}  \]

\item[(IAD)] There exists $\varrho>0$ such that for any bounded sets
$B_1,B_2\subset\Rd$ with dist($B_1,B_2$)$>\varrho$, the processes
$\{V_\omega(x):x\in B_1\}$ and $\{V_\omega(x): x\in B_2\}$ are independent.

\end{itemize}

In the case $H_0=H_B$, the unperturbed Landau Hamiltonian on $L^2(\mathbb
R^2)$
\begin{equation}\label{H_B1}
H_B = (-i\nabla-{\bf A})^2 \quad \hbox{ with } {\bf A}
=\frac{B}{2}(x_2,-x_1),
\end{equation}
 where {\bf A} is the vector potential and $B$ is the strength of
 the magnetic field, we ask ${\bf A}(x)\in L_{loc}^2(\mathbb R^2;\mathbb R^2)$
to satisfy the diamagnetic inequality so we can obtain trace estimates for the Landau
Hamiltonian from those of the Laplacian.

It follows that $H_\omega$ is a semibounded selfadjoint operator for $\mathbb
P$-a.e. $\omega$.  Moreover, the mapping $\omega\rightarrow H_\omega$ is
measurable for $\mathbb P$-a.e. $\omega$, we denote its spectrum by
$\sigma_{\omega}$.

In the usual setting for (ergodic) random Hamiltonians, $H_\omega$ satisfies a
\emph{covariance condition} with respect to the action of a family of unitary
(translation) operators $U_x$ , and its associated ergodic group of translations
$\tau_x$ on the probability space $\Omega$.  Throughout this paper we do not
make any assumption on the ergodicity of $H_\omega$, so this covariance
condition, \emph{a priori}, does not hold , i.e.  

\be\label{nonerg} H_{\tau_{\gamma}(\omega)}\neq U_\gamma H_{\omega}U_\gamma^*, \ee
which makes
$H_{\omega}$ a \emph{non-ergodic} random operator.
\bigskip

For the following assumption we need the notion of a \emph{finite volume
operator}, the restriction of $H_\omega$ to either an open box $\L_L(x)$
with Dirichlet boundary condition or to the closed box $\overline\L_L(x)$ with periodic boundary conditions. 
In this way, we obtain a well defined random operator $H_{\omega,x,L}$ acting on
$\Lp{\L_L(x)}$ defined by

\[ H_{\omega,x,L}=H_{0,x,L}+ \lambda V_{\omega,x,L}.  \]
we denote its spectrum by $\sigma_{\omega,x,L}$ and by $R_{\omega,
x,L}(z)=(H_{\omega,x,L}-z)^{-1}$ its resolvent operator.  We define the spectral projections $P_{\omega}(J)=\chi_J(H_{\omega})$ and $P_{\omega,x,L}(J)=\chi_J(H_{\omega,x,L})$ for
$J\subset\mathbb R$ a Borel set. When stressing the dependence on $\lambda$, it will be added to the subscript.

\begin{defn}\mbox{ }

\begin{itemize}
 \item[(UWE)]  We say that $H_\omega$ satisfies a uniform Wegner estimate with H\"older
exponent $s$ in an open interval $\mathcal J$, i.e., for every $E\in\mathcal J$
there exists a constant $Q_E$, bounded on compact subintervals of $\mathcal J$
and $0<s\leq 1$ such that 

\be\label{WE} 
\sup_{x\in\Rd}\mathbb E\{ \mbox{tr}(P_{\omega,x,L}(E- \eta, E+\eta) \}\leq Q_E\eta^s L^d,
\ee
for all $\eta>0$ and $L\in2\mathbb N$.  It satisfies a uniform Wegner estimate
at an energy $E$ if it satisfies a uniform Wegner estimate in an open interval
$\mathcal J$ such that $E\in\mathcal J$.
\end{itemize}
\end{defn}
\smallskip

\smallskip

To describe the dynamics, we consider the random
moment of order $p\geq 0$ at time $t$ for the time evolution in the
Hilbert-Schmidt norm, initially spatially localized in a square of
side one around $u\in\mathbb Z^2$ and localized in energy by the function
$\mathcal X \in C^\infty_{c,+}(\mathbb R)$, i.e.,

\be M_{u,\omega}(p,\mathcal X, t)= \Vert \langle X-u\rangle^{p/2}
e^{-itH_{\omega}}\mathcal X (H_{\omega})\chi_u
\Vert^2_2.\ee
\smallskip

We next consider its time average,

\be \mathcal M_{u,\omega}(p,\mathcal
X,T)=\frac{2}{T}\int_0^\infty e^{-2t/T}
M_{u,\omega}(p,\mathcal X,t) dt. \ee

%
%
%

\bigskip

\begin{defn}\mbox{ }
\begin{enumerate}
 \item[1.] We say that $H_{\omega}$ exhibits strong Hilbert-Schmidt (HS-) dynamical
localization in the open interval $I$ if for all $\mathcal X\in\mathcal
C^\infty_{c,+}(I)$ we have

\[ \displaystyle
\sup_{u\in\mathbb Z^2}\mathbb E\{ \sup_{t\in \mathbb R}  M_{u,\omega}(p,\mathcal X, t)\}<\infty \hspace{0.5cm}
\mbox{for all $p\geq 0$.}  \]
We say that $H_{\omega}$ exhibits strong Hilbert-Schmidt (HS-) dynamical
localization at an energy $E$ if there exists an  open interval $I$ with $E\in
I$, such that there is strong HS-dynamical localization in the open interval.

\smallskip

\item[2.] The strong insulator region for $H_\omega$ is defined as

\be \Sigma_{SI}= \{E\in\mathbb R: \mbox{$H_{\omega}$ exhibits strong HS-dynamical
localization at $E$} \}
\ee
Note that if there exists a $\delta>0$ such that $\dist(E,\sigma_\omega)>0$ for almost every $\omega$, then $E\in\Sigma_{SI}$.
\end{enumerate}
\end{defn}

\bigskip

As we shall see, the existence of such a region for random Schr\"odinger
operators is the consequence of the applicability of the Bootstrap MSA adapted to the
non ergodic setting (Theorem \ref{Bootstrap}).  
\bigskip

Given $\theta>0$, $E\in\mathbb R$, $x\in\mathbb Z^d$ and $L\in 6\mathbb N$, we
say that the box $\L_L(x)$ is \linebreak $(\theta,E)$-\emph{suitable}  for $H_\omega$
if $E\notin\sigma_{\omega,x,L}$ and

\[  \Vert \Gamma_{x,L}
R_{\omega,x,L}(E)\chi_{x,L/3}\Vert_{x,L}\leq \frac{1}{L^\theta} , \]
where $\Gamma_{x,L}=\chi_{\bar\L_{L-1}(x)\setminus\L_{L-3}(x)}$.  If
we replace the polynomial decay $1/L^\theta$ by $e^{-mL/2} $ we say that the box
$\L_L(x)$ is $(m,E)\emph{-regular}$ for $H_\omega$.

\smallskip
The following theorem is a reformulation of Theorem 3.4 and Corollary 3.10 \cite{GK1}
in a non ergodic setting, 

\begin{thm}\label{Bootstrap}
Let $H_\omega$ be a random Schr\"odinger operator satisfying a uniform Wegner
estimate in an open interval $\mathcal J$ with H\"older exponent $s$ and
assumptions (R), (IAD).  Given $\theta>d$, for each $E\in\mathcal J$ there
exists a finite scale $\mathcal L_\theta(E)=\mathcal L(\theta,E,Q_E,d,s)$,
bounded in compact subintervals of $\mathcal J$, such that if for $\mathcal L>\mathcal L_\theta(E)$ the following holds

\be\label{ILSE} \inf_{x\in\mathbb Z^d} \mathbb P\{\L_{\mathcal L}(x) \mbox{
is ($\theta,E$)-suitable}\}> 1-\frac{1}{841^d},
\ee
then there exists $\delta_0>0$ and $C_\zeta>0$ such that

\be \sup_{u\in\mathbb Z^d} \mathbb E \left(  \displaystyle\sup_{\Vert
f\Vert\leq 1} \Vert \chi_{x+u} f(H_\omega) P_\omega(I(\delta_0))\chi_u\Vert_2^2
\right)\leq C_\zeta e^{-|x|^\zeta}, \ee
for $0<\zeta <1$, where $I(\delta_0)=[E-\delta_0,E+\delta_0]$.  Moreover, $E\in
\Sigma_{SI}$ and we have the following properties,
\begin{itemize}
 \item[(SUDEC)] Summable uniform decay of eigenfunction correlations:  for a.e. $\omega\in\Omega$, the Hamiltonian $H_\omega$ has pure point spectrum in $I \subset\Sigma_{SI}$ with finite multiplicity.  Let $\{\epsilon_{n,\omega}\}_{n\in\mathbb N}$ be an enumeration of the distinct eigenvalues of $H_\omega$ in I.  Then for each $\zeta\in ]0,1[$ and $\epsilon >0$ we have, for every $x,u\in\mathbb Z^d$,

\be \norm{\chi_{x+u} \phi } \norm{\chi_{u} \varphi} \leq C_{I,\zeta,\epsilon,\omega} \norm{T^{-1}_u \phi} \norm{ T^{-1}_u \varphi } \langle x+u\rangle^{\frac{d+\epsilon}{2}} \langle u\rangle^{\frac{d+\epsilon}{2}} e^{-|x|^\zeta}, \ee 
for all $\phi,\varphi\ \in \mbox{ Ran } P_{\omega}(\{\epsilon_{n,\omega}\})$ (see Section \ref{Proof1}).


\item[(DFP)] Decay of the Fermi projections: for $E\in\Sigma_{SI}$ and for any $\zeta\in ]0,1[$ we have
\be \sup_{u\in\mathbb Z^d}\mathbb E\lbrace \norm{\chi_{x+u} P_{\omega}((-\infty, E])\chi_u}_2^2  \rbrace \leq C_{\zeta, \lambda,E} e^{-|x|^\zeta} \ee 
where the constant $C_{\zeta,E}$ is locally bounded in E.

\end{itemize}

\end{thm}

\begin{rem} \emph{The condition (\ref{ILSE}) is called  the
\emph{initial length scale estimate} (ILSE) of the Bootstrap MSA.  In practice is often useful to
prove the equivalent estimate \cite[Theorem 4.2]{GK3}:  For some $\theta>d$,
we have }

\be\label{ILSE2}\displaystyle\limsup_{L\rightarrow\infty} \inf_{x\in\mathbb Z^d}
\mathbb P\{\L_{\mathcal L}(x) \mbox{ is ($\theta,E$)-suitable}\}= 1.
\ee
\end{rem}

\begin{defn}
The multiscale analysis region for $H_{\omega}$ is defined as the set of
energies where we can perform the bootstrap MSA, i.e.

\begin{align}\Sigma_{MSA}= & \{E\in\mathbb R: \mbox{$H_{\omega}$ satisfies a
uniform Wegner estimate at $E$ and}\nonumber\\ &  \mbox{(ILSE) holds for some  $\mathcal L>\mathcal
L_\theta(E)$} \}
\end{align}
By Theorem \ref{Bootstrap}, we have  $\Sigma_{MSA}\subset\Sigma_{SI}$.

\end{defn}

\bigskip

We introduce the (lower) transport exponent in the annealed regime:

\be \beta(p,\mathcal X)= \displaystyle
\liminf_{T\rightarrow\infty}\frac{\log_+ \displaystyle\sup_u \mathbb E(\mathcal
M_{u,\omega}(p,\mathcal X, T))}{p\log T} \ee
for $p\geq 0$, $\mathcal X \in C^\infty_{c,+}(\mathbb R)$, where
$\log_+t= \max\{0,\log t\}$, and define the $p$-th local transport exponent
at the energy $E$,  by

\be \beta(p,E)=\displaystyle\inf_{I\ni E} \displaystyle
\sup_{\mathcal X\in C^\infty_{c,+}(I)} \beta(p,\mathcal
X), \ee where $I$ denotes an open interval. The exponents
$\beta(p,E)$ provide a measure of the rate of transport in wave
packets with spectral support near $E$.  Since they are increasing
in $p$, we define the local (lower) transport exponent
$\beta(E)$ by

\be \beta(E)=\displaystyle\lim_{p\rightarrow
\infty}\beta(p,E)=
\displaystyle\sup_{p>0}\beta(p,E). \ee

With the help of this transport rate we can define two complementary
sets in the energy axis for fixed $B>0$, $\lambda>0$, the region of
\emph{dynamical localization}

\be \Xi^{DL}=\{ E\in\mathbb
R:\mbox{ } \beta(E)=0  \},  \ee
also called the trivial transport region (TT) in \cite{GK3} and the region of
\emph{dynamical delocalization}

\be \Xi^{DD}=\{
E\in\mathbb R:\mbox{ } \beta(E)>0  \},
\ee
also called the weak metallic transport region (WMT), in \cite{GK3}.  Recalling Theorem \ref{Bootstrap}
we have that $\Sigma_{MSA}\subset\Sigma_{SI}\subset\Xi^{DL}$.

The following result is an improvement of  \cite[Theorem 2.11]{GK3} for the non
ergodic setting

\bigskip
\begin{thm}\label{Emom} Let $H_{\omega}$ be a Schr\"odinger operator
satisfying a uniform Wegner estimate with H\"older exponent $s$ in an open interval $\mathcal J$  and
assumptions (R), (IAD).  Let
$\mathcal X \in C^\infty_{c,+}(\mathbb R)$ with $\mathcal X \equiv
1$ on some open interval $J\subset \mathcal J$, $\alpha\geq0$ and
$p>p(\alpha,s):= 12 \frac{d}{s}+2\alpha\frac{d}{s}$.  If

\be \label{Emomentum} \displaystyle\liminf_{T\rightarrow\infty}
\sup_{u\in\mathbb Z^d} \frac{1}{T^{\alpha}}\mathbb E \left(\mathcal M_{u,\omega}(p,\mathcal X,T)\right)< \infty ,\ee
then $J\subset\Sigma_{MSA}$. In particular, it follows that
(\ref{Emomentum}) holds for any $p\geq0$.
\end{thm}

Moreover, we can extend this result to a quenched regime, a new feature in both ergodic and non-ergodic situations:

\bigskip

\begin{thm}\label{mom} Let $H_{\omega}$ be a Schr\"odinger operator
satisfying a uniform Wegner estimate with H\"older exponent $s$ in an open interval $\mathcal J$  and
assumptions (R), (IAD).  Let
$\mathcal X \in C^\infty_{c,+}(\mathbb R)$ with $\mathcal X \equiv
1$ on some open interval $J\subset \mathcal J$, $\alpha\geq0$ and
$p>p(\alpha,s):= 15 \frac{d}{s}+2\alpha\frac{d}{s}$.  If

\be\label{momentum} \displaystyle\liminf_{T\rightarrow\infty}
\sup_{u\in\mathbb Z^d} T^{\frac s d}\mathbb P (\mathcal M_{u,\omega}(p,\mathcal X,T)> T^\alpha)=0,\ee
then $J\subset\Sigma_{MSA}$. In particular, it follows that
(\ref{momentum}) holds for any $p\geq0$.
\end{thm}

\begin{rem}{\it If the moment increases almost surely at any other rate less than polynomial,
this implies in particular condition (\ref{momentum}) for some $\alpha>0$, and
the result follows. 

Moreover, if condition (2.28) in  \cite[Theorem 2.11]{GK3} holds for $\alpha>0$ and
$p>p(\alpha,s)+d$, then condition (\ref{momentum}) holds for
$\alpha'=\alpha+\delta$ and the same $p$, where $0<s/2<\delta<\frac{s(p-p(\alpha,s))}{2d}$ and
 $p>p(\alpha',s)$, since  by Chebyshev's inequality we have

\be T^{\frac s d}\sup_u\mathbb P (\mathcal M_{u,\omega}(p,\mathcal X,T)>
T^{\alpha'})\leq \frac{1}{T^{\alpha+\delta-s/2}}\sup_u\mathbb E (\mathcal
M_{u,\omega}(p,\mathcal X,T)) \mbox{ for all $T>0$. }
\ee
This also shows that (\ref{momentum}) is indeed a weaker condition than (\ref{Emomentum}).}

\end{rem}

\bigskip

By Theorem \ref{Emom} we have that $\Xi^{DL}\subset\Sigma_{MSA}$, so Theorems
2.8 and 2.10 of \cite{GK3} hold in our setting. Thus, the local transport exponent
$\beta(E)$ gives a characterization of the metal-insulator transport transition
for non ergodic models as for the usual ergodic setting.  Moreover, if we consider only the random moments in a quenched regime to behave asymptotically slow, we see the same behavior for the ergodic and non ergodic setting, in agreement with the annealed regime.

\bigskip

\section{Proof of Theorem \ref{Bootstrap}}\label{Proof1}

\subsection{Generalized eigenfunction expansion}
We have to construct a generalized eigenfunction expansion adapted to the non ergodic case.  Compared to \cite[Section 2.3]{GK3} we shall use a family of weighted spaces rather than just one in particular, using translations in $u\in\mathbb Z^2$ of the operator $T$
defined there and thus without using translation invariance in the proofs.

\bigskip

Let $T_u$ be the operator in $\H$ given by multiplication by
the function $(1+|x-u|^2)^\nu$, where $\nu>d/4$, $u\in\mathbb Z^2$.  We define
the weighted spaces $\mathcal
H_\pm^u$ as

\be\H_\pm^u=L^2(\Rd,(1+|x-u|^2)^{\pm 2\nu}
dx;\C). \ee

The sequilinear form 

\[\langle\phi_1,\phi_2\rangle_{\H_+^u,\H_-^u}=\int
\overline\phi_1 \phi_2 (x)dx \hspace{0.5cm}\mbox{for $\phi_1\in \H_+^u $
, $\phi_2\in\H_-^u$ }\]
makes $\H_+^u$ and $\H_-^u$
conjugates dual to each other and we denote by $\dagger$ the conjugation with
respect to this duality.  The natural injections  $\iota_+^u:\mathcal
H_+^u\rightarrow\H$ and $\iota_-^u:\H\rightarrow\H_-^u$
are continuous with dense range, with $(\iota_+^u)^\dagger=\iota_-^u$.  The
operators $T_{u,+}:\H_+^u\rightarrow\H$ and $T_{u,-}:\mathcal
H\rightarrow\H_-^u$ defined by $T_{u,+}=T_u\iota_+^u$,
$T_{u,-}=\iota_-^uT_u$ on $\mathcal D(T_u)$ are unitary with
$T_{u,-}=T_{u,+}^\dagger$.
Note that

\be \Vert \chi_{x,L}  \Vert_{\H, \mathcal
H^u_+}=\Vert \chi_{x,L}  \Vert_{\H^u_-, \mathcal
H}\leq C_{L,d, \nu}(1+|x-u|^2)^\nu, \ee
for all $x\in\Rd$ and $L>0$.

With this redefinition we can follow \cite{GK1}, restating assumption GEE for
non ergodic operators.  We consider a fixed open interval
$\mathcal I$  and we recall that
$P_{\omega}(J)= \chi_J(H_\omega)$ is the spectral projection of the operator
$H_\omega$ on a Borel set $ J\subset\mathbb R$.

\bigskip

(\textbf{UGEE})\emph{ For some $\nu>d/4$, the set $\mathcal D_{+}^{u,\omega}=\{
\phi\in\mathcal D(H_\omega)\cap\H_{+}^u:\mbox{ }
H_\omega\phi\in\H_+^u\}$ is dense in $\H_+$ and an operator core
for $H_\omega$ for $\mathbb P-$a.e. $\omega$ and all $u$.  There exists a
bounded function $f$, strictly positive on the spectrum of $H_\omega$ such
that,}  \}

\be \sup_u\mbox{ }\tr_\H\left( T_u^{-1}f(H_\omega)P_{\omega}(\mathcal
I)T_u^{-1}  \right)<\infty, \ee
\emph{for $\mathbb P-$a.e. $\omega$}.

If UGEE holds, for almost every
$\omega$ and all $u$ we have

\be \tr_\H \left(T_u^{-1}P_\omega(J\cap\mathcal I)T_u^{-1}
\right) <\infty, \ee
for all bounded sets $J$. Thus with probability one, for all $u$

\be\label{Defmu_u,omega} \mu_{u,\omega}(J)= \mbox{tr}_\H \left(
T_u^{-1}P_{\omega}(J\cap\mathcal I)T_u^{-1} \right) \ee
is a spectral measure for the restriction of $H_\omega$ to the
Hilbert space $P_{\omega}(\mathcal I)\H$, and for
every bounded set $J$,
\be \mu_{u,\omega}(J)<\infty.  \ee

Then, we have a generalized eigenfunction expansion as in  \cite[Section
2]{GK1}: for every $u$, there exists a
$\mu_{u,\omega}$-locally integrable function ${\bf P}_{u,\omega}(\tilde
\lambda)$ from $\mathbb R$ into $\mathcal T_1(\H^u_+,
\H^u_-)$, the space of trace class operators from $\mathcal
H^u_+$ to $\H^u_-$, with

\be {\bf P}_{u,\omega}(\tilde \lambda)={\bf P}_{u,\omega}(\tilde
\lambda)^\dagger
\ee
and 
\be \tr_\H \left(T_{u,-}^{-1}{\bf P}_{u,\omega}(\tilde \lambda)T_{u,+}^{-1}
\right) =1 \hspace{0.5cm}\mbox{for $\mu_{u,\omega}$-a.e. $\tilde\lambda$,}
\ee
such that 

\be\label{ProjectionPu} \iota_-^u P_\omega(J\cap\mathcal I)\iota_+^u=\int_J {\bf P}_{u,\omega}(\tilde
\lambda)d\mu_{u,\omega}(\tilde\lambda) \hspace{0.5cm}\mbox{for bounded Borel sets
$J$}
\ee
where the integral is the Bochner integral of $\mathcal T_1(\H^u_+,
\H^u_-)$-valued functions.

The following (a restatement of assumption SGEE), is 
a stronger version of UGEE:

 \textbf{(USGEE)} \emph{ We have that (UGEE) holds with } 
\be \label{USGEE} \sup_u \mathbb E \left( [\tr_\H\left(
T_u^{-1}f(H_\omega)P_{\omega}(\mathcal I)T_u^{-1}\right)]^2 \right)<\infty  .\ee

So for every bounded set $J$,

\be \sup_u \mathbb E(\mu_{u,\omega}(J)^2)<\infty.  \ee

\subsection{Kernel Decay and Dynamical Localization}

Following the arguments in \cite{GK1} for ergodic operators, we can show that
HS-strong dynamical localization is a consequence of the applicability of the
Bootstrap MSA for the non ergodic setting  (\cite[Theorem 3.4]{GK1} with the
stronger initial ILSE (\ref{ILSE}) instead of the original one).

We can restate Lemma 2.5 and Lemma 4.1 \cite{GK1} as follows,
extending the proofs to our new definitions,

\begin{lem}\label{Pu} Let $H_\omega$ be a random operator satisfying assumption
GEE.  We have with probability one, for all $u$, that
for $\mu_{u,\omega}$-almost every $\tilde\lambda$,

\be \Vert \chi_x {\bf P}_{u,\omega}(\tilde\lambda)\chi_y \Vert_1\leq C
(1+|x-u|^2)^\nu(1+|y-u|^2)^\nu \ee
for all $x,y \in\Rd$, with $C$ a finite constant independent
of $\tilde\lambda,\omega$ and $u$.

Suppose, moreover, that assumption EDI in \cite{GK1} is satisfied in some
compact
interval $I_0\subset\mathcal I$.  Given $I\subset I_0$, $m>0$, $L\in6 \mathbb N$
and
$x,y\in \mathbb Z^d$, if $\omega\in R(m,L,I,x,y)$, with $R(m,L,I,x,y)$ defined
as in (\ref{R(x,y)}), then

\be \Vert \chi_x {\bf P}_{u,\omega}(\tilde\lambda)\chi_y \Vert_2\leq C
e^{-mL/4} (1+|x-u|^2)^\nu(1+|y-u|^2)^\nu,  \ee
for $\mu_{u,\omega}$-almost all $\tilde\lambda\in I$, with
$C=C(m,d,\nu,\tilde\gamma_{I_0})$, where $\tilde\gamma_{I_0}$ is the
constant on assumption EDI.
\end{lem}

\bigskip

\begin{proof}[Proof of Theorem \ref{Bootstrap}]
To apply the MSA in the non ergodic case we first need to verify for an operator
 satisfying only properties (R), (IAD) and (UWE), the standard assumptions (SLI),
(EDI) \cite{GK1}, plus (UNE) and (USGEE), which are stronger assumptions than those stated in the
mentioned article.  

\smallskip

As for (SLI) and (EDI), these are deterministic assumptions that hold for each
$\omega \in \Omega$ and their proof, done in \cite[Appendix A]{GK3}, relies on
property (R), with no use of ergodicity.  In the same appendix we see that
assumption (NE) is uniform on cubes centered in $x\in\Rd$ and relies on
property (R) so it holds in our more general setting.  The same is true for
 \cite[Lemma A.3]{GK3}, and can be extended in an analog way to the case $H_0=H_B$  \cite[Section 2.1]{BGKS}, proving the first part of (USGEE) (and (UGEE)).

As for the trace estimate (\ref{USGEE}) , for the case $H_0=-\Delta$ it follows
from  \cite[Lemma A.4]{GK3} and  \cite[Theorem 1.1]{KKS}, taking $V=\langle
X-u\rangle^{-2\nu}$ there, the result being uniform in $u$.  It can be extended
to the case $H_0=H_B$ as in  \cite[Proposition 2.1]{BGKS}.

To obtain the basic result of MSA  \cite[Theorem 3.4]{GK1} we need conditions
(IAD), (SLI), (UNE) and (UWE) to follow an analog iteration procedure.  Recall that
in their article, Germinet and Klein take two versions of MSA by Figotin and
Klein, improve their estimates yielding other two MSA and then
bootstrapping them to obtain the strongest result out of the weakest hypothesis,
so in order to extend this results to the non ergodic setting we reformulate this
methods. Each step consists of a purely geometric deterministic part where we
use SLI, and therefore it does not depend on the placement of the boxes were we
perform the procedure, and a probabilistic part, where we use (UWE) instead of
(WE) to obtain an estimate on the probability of having bad events, in a stronger
sense than the usual, that is, uniform with respect to the placement of the box in space.

 We begin with the single energy multiscale analyses, Theorems 5.1 and 5.6
\cite{GK1}, which in our non-ergodic setting consists in estimating the decay
of 

\be p_L=\displaystyle\sup_{x\in \mathbb Z^d} p_{x,L},\ee
 where

\be p_{x,L}= \mathbb P\{ \L_L(x) \mbox{ is bad} \}
\ee
(here a box is bad if it is not $(\theta,E)$-\emph{suitable} for
$H_\omega$).  In the ergodic case we need only to consider $p_{0,L}$. 
Hypothesis (\ref{ILSE}) ensures we can follow the same iteration procedure in
all boxes centered in $x\in \mathbb Z^d$, where $p_{0,L}$ is thus replaced by
$p_L$.  We use properties SLI and (UWE) instead of WE, and the deterministic
arguments remain the same, since  they do not depend on the location of the box. 
Considering a H\"older exponent $s$ in WE implies that the choice of the initial
length scale will also depend on $s$.

Next we consider the energy interval multiscale analyses, Theorems 5.2 and 5.7
\cite{GK1}, which in our general setting consists in estimating 

\be\tilde p_L=\displaystyle\sup_{x,y\in \mathbb Z^d \atop |x-y|>L+\varrho}
\tilde p_{x,y,L},\ee 
with

\be \tilde p_{x,y,L}= \mathbb P\{R(m,L,I(\delta_0),x,y)^c\}\ee
where $I(\delta_0)=[E-\delta_0,E+\delta_0]$, for some $\delta_0>0$
and 

\be \label{R(x,y)} R(m,L,I(\delta_0),x,y)=\{\omega: \mbox{ for every $E\in
I(\delta_0)$,  }\L_L(x) \mbox{ or }\L_L(y) \mbox{ is good}\}\ee
(here a box is good if it is $(m,E)$-\emph{regular} for $H_\omega$,
with $m$ to be specified later).  In the ergodic case it suffices to consider
$\tilde p_{x,y,L}$.  We can thus follow the original iteration procedure on this
estimate, replacing $\tilde p_{x,y,L}$ by $\tilde p_L$, obtaining an analog of
 \cite[Eq. 3.4]{GK1}, i.e., there exists $\delta_0>0$ such that given any
$\zeta$, $0<\zeta <1$ there is a length scale $L_0<\infty$ and a mass
$m_\zeta=m(\zeta,L_0)>0$ such that if we set $L_{k+1}=[L_k^\alpha]_{6\mathbb
N}$, $0<\alpha<\zeta^{-1}$ ,$k=0,1,2,...$ we have

\be\label{outputMSA} \displaystyle\inf_{x,y\in \mathbb Z^d \atop
|x-y|>L+\varrho}\mathbb P\{R(m_\zeta,L_k,I(\delta_0),x,y) \} \geq
1-e^{-L^\zeta_k} .\ee

\smallskip

To derive results on the spectrum and the dynamics of the operator from this estimate we need to consider also conditions
EDI and USGEE.  Thus, with Lemma \ref{Pu} in hand, (\ref{outputMSA}) and USGEE
we can follow the proof of
 \cite[Theorem 3.8]{GK1} with minor modifications. We want to show that if (\ref{outputMSA}) holds we have that for any
$0<\zeta<1$, there is a finite constant $C_\zeta$ such that

\be\label{KerDecay} \sup_u \mathbb E \left(\displaystyle\sup_{\Vert
f\Vert\leq 1} \Vert \chi_{x+u} f(H_\omega) P_\omega(I(\delta_0))\chi_u\Vert_2^2
\right)\leq C_\zeta e^{-|x|^\zeta}, \ee
For this, we consider the pair of points $x,y$ as the pair $x+u,u$, and fix $x\in\mathbb Z^d$ and $k$ such that $L_{k+1}+\varrho>|x|>L_k+\varrho$.  We split the expectation in (\ref{KerDecay}) in two parts: the first one over the set
$R(m_\zeta,L_k,I(\delta_0),x+u,u)$ and the second one over its complement, which has probability
less than $e^{-L^\zeta_k}$, uniformly in $u$, by (\ref{outputMSA}).  We follow the arguments in  \cite[Eq. 4.8-4.13]{GK1}.
By (\ref{ProjectionPu}) and Lemma \ref{Pu} we can write, for a positive constant $C_1$,

\be \displaystyle\sup_{\Vert
f\Vert\leq 1} \Vert \chi_{x+u} f(H_\omega) P_\omega(I(\delta_0))\chi_u\Vert_2 \leq C_1e^{-L_k^\zeta}\mu_{u,\omega}(I). \ee 
This implies,
\begin{align}\label{exp1} & \sup_u\mathbb E \left( \displaystyle\sup_{\Vert
f\Vert\leq 1} \Vert \chi_{x+u} f(H_\omega) P_\omega(I(\delta_0))\chi_u\Vert_2^2
; R(m_\zeta,L_k,I(\delta_0),x+u,u)\right) \nonumber\\ & \leq  C_1^2\sup_u\mathbb E\{(\mu_{u,\omega}(I(\delta_0)))^2\}e^{-2L_k^\zeta} \end{align}
As for the expectation over $R(m_\zeta,L_k,I(\delta_0),x+u,u)^c$, (\ref{outputMSA}) implies that \[\displaystyle\sup_u\mathbb P(R(m_\zeta,L_k,I(\delta_0),x+u,u)^c)<e^{-L^\zeta_k}\]
this yields,

\begin{align}\label{exp2} & \sup_u\mathbb E \left(\displaystyle\sup_{\Vert
f\Vert\leq 1} \Vert \chi_{x+u} f(H_\omega) P_\omega(I(\delta_0))\chi_u\Vert_2^2
; R(m_\zeta,L_k,I(\delta_0),x+u,u)^c\right)\nonumber\\ & \leq  4^\nu\sup_u \mathbb E\{(\mu_{u,\omega}(I(\delta_0)))^2\}^{\frac{1}{2}}e^{-\frac{1}{2}L_k^\zeta} \end{align}
where we use the fact that by (\ref{Defmu_u,omega}) we can write

\be \Vert \chi_{x+u} f(H_\omega) P_\omega(I(\delta_0))\chi_u\Vert_2^2 \leq \norm{f}^2\norm{P_\omega(I(\delta_0))\chi_u}_2^2 \leq C\norm{f}\mu_{u,\omega}(I(\delta_0))  \ee 
Combining  (\ref{exp1}) and (\ref{exp2}), using USGEE we obtain the desired decay, namely (\ref{KerDecay}).

\smallskip

Now we can prove a strong version of dynamical localization as in  \cite[Corollary 3.10]{GK1}.  Notice that, if $p>2$

\begin{align} 
\langle{X-u}\rangle^p= \sum_{x\in\mathbb Z^d} (1+\Vert y-u\Vert^2)^{p/2}\chi_x(y)
&\leq  C_d \sum_{x\in\mathbb Z^d} (1+\Vert x-u\Vert^2)^{p/2}\chi_x(y)\nonumber\\
& = C_d \sum_{x\in\mathbb Z^d} (1+\Vert x\Vert^2)^{p/2}\chi_{x+u}(y),\end{align}
so we have, 

\begin{align}
\Vert \langle{X-u}\rangle^{p/2} &f(H_{\omega})P_\omega(I(\delta_0))\chi_u
\Vert_2^2 \nonumber\\ 
& = \tr [\chi_u f(H_{\omega})P_\omega(I(\delta_0))\langle{X-u}\rangle^p 
P_\omega(I(\delta_0))f(H_{\omega})\chi_u]\nonumber\\
& \leq C_d \sum_{x\in\mathbb Z^d} (1+\Vert x\Vert^2)^{p/2} \tr [\chi_u
f(H_{\omega})P_\omega(I(\delta_0)) \chi_{x+u}
P_\omega(I(\delta_0))f(H_{\omega})\chi_u]\nonumber\\
& = C_d\sum_{x\in\mathbb Z^d} (1+\Vert x\Vert^2)^{p/2}\Vert
\chi_{x+u}f(H_{\omega})P_\omega(I(\delta_0))\chi_u \Vert_2^2 
\end{align}
Taking the expectation and then the supremum over $u\in\mathbb Z^2$, by (\ref{KerDecay}) we obtain
strong HS-dynamical localization in the energy interval $I(\delta_0)$.

Following the proof of \cite[Corollary 3]{GK4}, after adapting \cite[Theorem 1]{GK4} to our setting 
 we obtain the summable uniform decay of eigenfunction correlations SUDEC.  As for property DFP, it is a consequence of (\ref{KerDecay}) combined with \cite[Theorem 1.4]{BGK}, which is a deterministic result also valid in our setting, in the lines of \cite[Theorem 3]{GK4} .

\end{proof}


\section{Proofs of Theorems \ref{Emom} and \ref{mom}} \label{Proofs23}

Here we can proceed as in \cite{GK3}. First we state the following Lemma, which
is an intermediate result in the proof of \cite[Lemma 6.4]{GK3}, adapted to the
(UWE) with H\"older exponent $s$.  We consider a cube $\Lambda_L(x)$ with arbitrary $x$ so we omit it from the notation.

\begin{lem}\label{lemproba}
Let $H_\omega$ be a random Schr\"odinger operator satisfying a
uniform Wegner estimate in an open interval $\mathcal I$, with Wegner
constant $Q_E$ and H\"older exponent $s$. Let $p_0>0$ and $\gamma>d$.  For each
$E\in\mathcal
I$, there exists $\mathcal L=\mathcal L(d,E,Q_E,\gamma, p_0,s)$
bounded on compact subsets of $\mathcal I$, such that, given
$L\in2\mathbb N$ with $L\geq \mathcal L$, and subsets $B_1$ and
$B_2$ of $\L_L$(not necessarily disjoint) with $B_1\subset
\L_{L-5/2}$ and $\bar\L_{L-1}\setminus\L_{L-3}\subset B_2$, then
for each
$a>0$ and $0<\epsilon\leq 1$ we have

 \be\label{lemproba1} \mathbb P\left(\Vert \chi_2
R_{\omega,L}(E+i\epsilon)\chi_1\Vert_L>\frac{a}{4}\right)\leq
\mathbb P\left(\Vert \chi_2R_\omega(E+i\epsilon)\chi_1  \Vert
>\frac{a}{L^\gamma}\right) + \frac{p_0}{10}, \ee
and

\be\label{lemproba2} \mathbb P\left(\Vert \chi_2
R_{\omega,L}(E)\chi_1\Vert_L>\frac{a}{2}\right)\leq \mathbb
P\left(\Vert \chi_2R_\omega(E+i\epsilon)\chi_1  \Vert
>\frac{a}{L^\gamma}\right) +
Q_E\left(\frac{4\epsilon}{a}\right)^{s/2}L^d+ \frac{p_0}{10},\ee
where $\chi_i$ stands for $\chi_{B_i}$, $i=1,2$.
\end{lem}

\bigskip

\begin{proof}[Proof of Theorem \ref{Emom}]

By the same arguments used in \cite[Theorem 4.2]{GK3}, it suffices to show that,
under condition (\ref{momentum}), for each
$E\in J$ there is some $\theta>d/s$ such that 

\be\label{CImsa} \displaystyle \limsup_{L\rightarrow\infty}
\inf_{y\in\mathbb Z^d} \mathbb P \left( \Vert \Gamma_{y,L}
R_{\omega,y,L}(E)\chi_{y,L/3} \Vert_{y,L} \leq \frac{1}{L^\theta}
\right)=1,  \ee
i.e. the starting condition for the bootstrap MSA, (\ref{ILSE}), in its strong
version, holds at some finite scale $L>\mathcal L_\theta(E)$.

Let $E\in J$, $\theta>d/s$ and $L\in 6\mathbb N$.  We start by
estimating

\be P_{E,L}:=\displaystyle\sup_y\mathbb P \left( \Vert \Gamma_{y,L}
R_{\omega,y,L}(E)\chi_{y,L/3} \Vert_{y,L} > \frac{1}{L^\theta}
\right).\ee
We decompose as in \cite[Eq. 6.26-6.28]{GK3}, using 
\[ \chi_{y,L}=\chi_{y,2L/3}+\chi_{y,L\backslash 2L/3},\mbox{\hspace{.5cm}where  
} \chi_{y,L\backslash2L/3}=\chi_{y,\L_L\setminus\L_{2L/3}}  \]
so (for simplicity we omit the subscript $y$ from the norm)

\bea P_{E,L} & \leq  & \displaystyle\sup_y \mathbb P\left(
\frac{1}{4L^\theta} < \Vert \Gamma_{y,L}
R_{\omega,L}(E+i\epsilon)\chi_{y,L/3}\Vert_L \right)\\
& + & \displaystyle\sup_y \mathbb P \left( \frac{1}{2L^\theta} <
\epsilon \Vert R_{\omega,L}(E+i\epsilon) \Vert_L\Vert
\Gamma_{y,L} R_{\omega,L}(E)\chi_{y,2L/3} \Vert_L \right) \\
& + & \displaystyle\sup_y \mathbb P \left( \frac{1}{4L^\theta}<
\epsilon\Vert R_{\omega,L}(E)\Vert_L \Vert\chi_{y,L\backslash2L/3}
R_{\omega,L}(E+i\epsilon)\chi_{y,L/3} \Vert_L \right).
 \eea

To estimate the first term we use (\ref{lemproba1}) with
$a=L^{-\theta}$.  As for the rest, we use (\ref{lemproba2}) and
(\ref{lemproba1}), respectively, with $a=1$, plus the uniform Wegner
estimate. We obtain

\bea P_{E,L} & \leq  & \label{prob1} \displaystyle\sup_y \mathbb P\left(
\frac{1}{L^{\theta+\gamma}} < \Vert \Gamma_{y,L}
R_\omega(E+i\epsilon)\chi_{y,L/3}\Vert \right)\\
& + & \label{prob2} \displaystyle\sup_y \mathbb P \left( \frac{1}{L^\gamma} <
\Vert
\Gamma_{y,L} R_\omega(E+i\epsilon)\chi_{y,2L/3} \Vert \right) \\
& + & \label{prob3} \displaystyle\sup_y \mathbb P \left(\frac{1}{L^\gamma}<
\Vert\chi_{y,L\backslash2L/3}
R_\omega(E+i\epsilon)\chi_{y,L/3} \Vert\right)\\
& + & \label{prob4} Q_I(4\epsilon)^{s/2} L^d + 2Q_I \epsilon^{s}L^{\theta s+d} +
\frac{3p_0}{10}, \eea
for $L>\mathcal L$, with $\mathcal L$ as in Lemma \ref{lemproba}, where $\gamma>d/s$,
$0<\epsilon\leq 1$, $0<p_0<1$ and $Q_I=\displaystyle\sup_{E\in
I}Q_E<\infty$ .   Set

\be\label{L}
L=L(I,\epsilon):=\left[\left(\frac{p_0}{20Q_I\epsilon^s}\right)^{1/(\theta
s+d)}\right]_{6\mathbb
N},\ee
so that
\[Q_I(4\epsilon)^{s/2} L^d\leq \frac{p_0}{10} \mbox{\hspace{0.3cm}
and\hspace{0.3cm} }2Q_I \epsilon^{s}L^{\theta s+d}\leq \frac{p_0}{10}.\]

We first estimate,

\be\displaystyle\sup_y \mathbb P\left( \frac{1}{L^{\theta+\gamma}} <
\Vert \Gamma_{y,L} R_\omega(E+i\epsilon)\chi_{y,L/3}\Vert \right).
\ee
To do this, we decompose the norm using the function $\mathcal X(H_\omega)$ that localizes in energy , yielding

\be \label{probXi1} \displaystyle\sup_y \mathbb P\left( \frac{1}{2L^{\theta+\gamma}}
< \Vert \Gamma_{y,L} R_\omega(E+i\epsilon)\mathcal
X(H_\omega)\chi_{y,L/3}\Vert \right)\ee

\be \label{probXi2} + \displaystyle\sup_y \mathbb P\left(
\frac{1}{2L^{\theta+\gamma}} < \Vert \Gamma_{y,L}
R_\omega(E+i\epsilon)(1-\mathcal X(H_\omega))\chi_{y,L/3}\Vert
\right). \ee
For the second term we use Chebyshev's inequality and follow \cite[Eq.
6.32 - 6.34]{GK3}, so we can bound it by $p_0/12$.

Estimating in the same way the terms (\ref{prob2}) and (\ref{prob3}) we obtain that for $L$ big enough,

\begin{align} P_{E,L} \leq & \displaystyle\sup_y \mathbb P\left( \frac{1}{2L^{\theta+\gamma}}
< \Vert \Gamma_{y,L} R_\omega(E+i\epsilon)\mathcal
X(H_\omega)\chi_{y,L/3}\Vert \right) \label{probXi} \\
& + \displaystyle\sup_y \mathbb P \left( \frac{1}{2L^\gamma} <
\Vert \Gamma_{y,L} R_\omega(E+i\epsilon)\mathcal
X(H_\omega)\chi_{y,2L/3} \Vert \right) \label{probXii}\\
& + \displaystyle\sup_y \mathbb P \left(\frac{1}{2L^\gamma}<
\Vert\chi_{y,L\backslash2L/3}
R_\omega(E+i\epsilon)\mathcal
X(H_\omega)\chi_{y,L/3} \Vert\right) + \frac{3p_0}{4}. \label{probXiii}
 \end{align} 
As for the first term,

\begin{align}
& \mathbb P \left(
\frac{1}{2L^{\theta+\gamma}} < \norm{\Gamma_{y,L} R_\omega(E+i\epsilon)\mathcal
X(H_\omega)\chi_{y,L/3} }\right)  \nonumber \\ & \hspace{3cm} \leq   2L^{\theta+\gamma}
 \mathbb E \left(\norm{\Gamma_{y,L} R_\omega(E+i\epsilon)\mathcal
X(H_\omega)\chi_{y,L/3} }\right)\\
 & \hspace{3cm} \leq  2L^{\theta+\gamma} \displaystyle \sum_{u\in \tilde\L_{L/3}(y)} \mathbb E \left(\norm{\Gamma_{y,L} R_\omega(E+i\epsilon)\mathcal
X(H_\omega)\chi_u}\right).
\end{align}
For any $u$ fixed, given a compact subinterval $I\subset J$ and $M>0$  we set :

\[A_{u,M,I,\epsilon}=\left\{ E\in I:\mbox{
}\mathbb E \left( \norm{\angles{X-u}^{p/2}R_\omega(E+i\epsilon)\X(\Homega)\chi_u}_2^2\right)\leq
M\epsilon^{-(\alpha+1)} \right\}.  \]
We have, taking $T=\epsilon^{-1}$ and using  \cite[Lemma 6.3]{GK3}

\begin{align} |I\setminus A_{u,M, I,\epsilon}|& \leq \frac{1}{M\epsilon^{-(\alpha+1)}}\int_\R \mathbb E \left(\norm{\angles{X-u}^{p/2}R_\omega(E+i\epsilon)\X(\Homega)\chi_u}_2^2 \right) dE\nonumber\\
 & = \frac{2\pi}{M T^{\alpha+1}}\int_0^\infty e^{-2t/T} \mathbb E \left(\norm{\angles{X-u}^{p/2}e^{-it\Homega}\X(\Homega)\chi_u}_2^2 \right)dt\nonumber\\
 & \leq  \frac{\pi}{M T^\alpha} \sup_u \mathbb E\left(\mathcal M_{u,\omega}(p,\X,T)\right).\nonumber\\
\end{align}
\begin{rem}{\it
Notice that the analogous sets $A_{k,I,M}$ in the proof  \cite[Theorem 2.11]{GK3} do not work in the non ergodic setting, so we need to consider a family of sets $A_{u,M, I,\epsilon}$, indexed by $u$.}
\end{rem}

By hypothesis \ref{Emomentum} we can pick a sequence $T_k\rightarrow\infty$ such that for $k$ big enough, we have $\displaystyle\sup_u \mathbb E\left(\mathcal M_{u,\omega}(p,\X,T_k)\right) <C T_k^\alpha$, then for the corresponding sequence $\epsilon_k \rightarrow 0^+$  we have
\smallskip

\be\label{Acomp1} |I\setminus A_{u,M, I,\epsilon_k}| \leq \frac{C}{M}. \ee
Notice that this  bound is uniform in $u$.

Thus, for an $E\in I$ fixed and $\epsilon_k= T_k^{-1}$, either $E\in A_{u, I,M,\epsilon_k}$ in which case we have,  

\begin{align}\label{ExpEst}\mathbb E \left(\norm{\Gamma_{y,L_k} R_\omega(E+i\epsilon_k)\mathcal
X(H_\omega)\chi_{u} }\right) & \leq C_{p,d}L_k^{-p/2}\mathbb E\left( \norm{\angles{X-u}^{p/2}R_\omega(E+i\epsilon_k)\X(\Homega)\chi_u}_2\right)\nonumber\\
& \leq C_{p,d}L_k^{-p/2}\mathbb E\left( \norm{\angles{X-u}^{p/2}R_\omega(E+i\epsilon_k)\X(\Homega)\chi_u}_2^2\right)^{1/2}\nonumber\\
& \leq C_{p,d}L_k^{-p/2}M^{1/2}\epsilon_k^{-(\alpha+1)/2},
 \end{align}
where we write $L_k=L(I,\epsilon_k)$, or else, $E\in I\setminus A_{u, M,I,\epsilon_k} $, so by (\ref{Acomp1}) there exists $E_{u}\in A_{u, I,M,\epsilon_k}$ such that 

\[ |E-E_{u}|\leq \frac{ C}{ M}\]
and so, by the resolvent identity and the definition of $A_{u,M,I, \epsilon}$,

\begin{align} \mathbb E \left( \norm{\Gamma_{y,L_k} R_\omega(E+i\epsilon_k)\mathcal
X(H_\omega)\chi_{u} } \right) & \leq 
\mathbb E \left( \norm{\Gamma_{y,L_k} R_\omega(E_{u}+i\epsilon_k)\mathcal
X(H_\omega)\chi_{u} } \right) \nonumber\\ & + |E-E_{u}| \mathbb E \left( \Vert
R_{\omega}(E+i\epsilon_k) \Vert \Vert
R_{\omega}(E_{u}+i\epsilon_k) \Vert \right) \nonumber\\
&\leq C_{p,d}L_k^{-p/2}M^{1/2}\epsilon_k^{-(\alpha+1)/2}  +
\frac{C}{M\epsilon_k^2}.
\end{align}
Therefore, 

\begin{align}\label{proba1}
\mathbb P \left(
\frac{1}{2L_k^{\theta+\gamma}}  <
\norm{\Gamma_{y,L_k} R_\omega(E+i\epsilon_k)\mathcal
X(H_\omega)\chi_{y,L_k/3} }\right) & \leq C'_{p,d} L_k^{\theta+\gamma-p/2+d}M^{1/2}\epsilon_k^{-(\alpha+1)} \nonumber\\
& + C''_{p,d}\frac{L_k^{\theta+\gamma+d}}{M\epsilon_k^2}.\nonumber\\
\end{align}

\smallskip

The remaining terms (\ref{probXii}) and (\ref{probXiii}) are estimated in the same way, using the fact that $\mbox{dist}(\bar\L_{L-1}\setminus\L_{L-3}, \Lambda_{\frac {2L}{3}})\geq \frac L 3 - \frac 3 2$ and $\mbox{dist}(\Lambda_{L\setminus \frac{2L}{3}}, \Lambda_{\frac {2L}{3}})\geq \frac L 6$.  For these terms we obtain an estimate as (\ref{proba1}) with constants $C_{p,d}'^{(2)}$,$C_{p,d}''^{(2)}$  and $C_{p,d}'^{(3)}$, $C_{p,d}''^{(3)}$, respectively, and with no $\theta$ in the exponent of $L$.  Denote by $C_{p,d}$ the maximal constant, and  since $L^\theta<L^{\theta+\gamma}$, the estimate on (\ref{proba1}) using $C_{p,d}$ will imply the same estimate on (\ref{probXii}) and on (\ref{probXiii}).

\smallskip

Now, for $p$ such that $p>p'(\alpha,s)=\alpha \frac{2d}{s}+12\frac{d}{s}$,  we can find $\theta,\gamma>d/s$ for which

\be p > 5\theta+3\gamma +2d + (\alpha+1)(\theta s+d)/s,\ee
so if we set 

\be M= L_k^{3\theta+\gamma}, \ee
and recall

\be\epsilon_k
^{-(\alpha+1)/2}= C_{p_0,Q_I}L_k^{(\alpha+1)(\theta s+d)/2s}, \quad  \epsilon_k
^{-2}=C_{p_0,Q_I}' L_k^{-2(\theta s +d)/s}. \ee
we obtain, for $k$ big enough depending on $d,I,p,\alpha,\theta,\gamma,s,p_0, Q_I$, 

\begin{align}
C'_{p,d} L_k^{\theta+\gamma-p/2+d}M^{1/2}\epsilon_k^{-(\alpha+1)} < p_0/24
\end{align}
and
\be  C''_{p,d}\frac{L_k^{\theta+\gamma+d}}{M\epsilon_k^2} < p_0/24,
\ee
so there exists a sequence $L_k\rightarrow \infty$ such that for $k$ big enough,

\be \mathbb P \left(
\frac{1}{2L_k^{\theta+\gamma}}  <
\norm{\Gamma_{y,L_k} R_\omega(E+i\epsilon_k)\mathcal
X(H_\omega)\chi_{y,L_k/3} }\right) < \frac{p_0}{12}.  \ee 
The same argument shows that the terms (\ref{probXii}) and (\ref{probXiii}) are smaller than $p_0/12$, for $k$ big enough.

\smallskip

Inserting this in (\ref{probXi})-(\ref{probXiii}) we see that

\be\displaystyle \limsup_{k\rightarrow \infty}\sup_y\mathbb P \left(
\frac{1}{L_k^\theta} <
\Vert\Gamma_{y,L_k}R_{\omega,y,L_k}(E)\chi_{y,L_k/3} \Vert_{L_k}
\right)\leq p_0, \ee
Since $0<p_0<1$ is arbitrary, we conclude that (\ref{CImsa}) holds
for each $E\in I$.

\end{proof}

\begin{proof}[Proof of Theorem \ref{mom}]

From equation (\ref{CImsa}) to equation (\ref{probXiii}) the previous proof remains valid in the current setting.  We will only estimate (\ref{probXi}), since the remaining terms (\ref{probXii}) and (\ref{probXiii}) can be estimated in the same way.  Notice that

\[\P\left(\frac{1}{2L^{\theta+\gamma}} < \norm{\Gamma_{y,L} R_\omega(E+i\epsilon)\mathcal
X(H_\omega)\chi_{y,L/3} } \right)\]
\begin{align}\leq & \P\left(\frac{1}{2L^{\theta+\gamma}}<\sum_{u\in \tilde\L_{L/3}(y)}\norm{\Gamma_{y,L} R_\omega(E+i\epsilon)\mathcal
X(H_\omega)\chi_{u} }
\right)\nonumber\\
\leq & \sum_{u\in \tilde\L_{L/3}(y)}\P\left(\frac{1}{2L^{\theta+\gamma+d}}< \Vert
\Gamma_{y,L} R_\omega(E+i\epsilon)\mathcal
X(H_\omega)\chi_{u}\Vert\right) \end{align}

\smallskip
To estimate the r.h.s of the last inequality, the following following lemma is crucial,

\begin{lem}\label{deterministiclemma}
There exists $\mathcal L =\mathcal L(I,p,\theta,\gamma,d,\alpha,s,p_0, Q_I)$ such that for any $u\in \tilde\L_{L/3}(y)$ with $L=L(I,\epsilon)$ as in (\ref{L}), $L\geq \mathcal L$ and  $E\in I$ fixed, if 

\be p>p(\theta,\gamma,d,\alpha,s):= \alpha\frac{(\theta s+d)}{s}+9\theta+3\gamma+2d+\frac d s\ee
then, for $T=\epsilon^{-1}$,

\be \left\{ \omega:\mbox{ }\Vert
\Gamma_{y,L} R_\omega(E+i\epsilon)\mathcal
X(H_\omega)\chi_{u}\Vert>\frac{1}{2L^{\theta+\gamma+d}}\right\}\subset \left\{\omega:\mbox{ }\mathcal M_{u,\omega}(p,\mathcal X,T)>T^\alpha\right\}.  \ee
\end{lem}

Now, if $p>p(\alpha,s):= 15 \frac{d}{s}+2\alpha\frac{d}{s}$, then there exist $\theta,\gamma >d/s$ such that $p> p(\theta,\gamma,d,\alpha,s)>p(\alpha,s)$ so 
Lemma \ref{deterministiclemma} holds yielding, for $L=L(I,\epsilon)$ as in (\ref{L}) big enough,

\begin{align} \P\left(\frac{1}{2L^{\theta+\gamma}}< \norm{\Gamma_{y,L} R_\omega(E+i\epsilon)\mathcal
X(H_\omega)\chi_{y,L/3} }  \right) \leq C_{p_0,Q_I} T^{\frac s 2}\sup_u\P(\mathcal M_{u,\omega}(p,\mathcal X,T)>T^\alpha)
\end{align}
where $C_{p_0,Q_I}$ comes from $L^d=C_{p_0,Q_I}T^{\frac s 2}$, by (\ref{L}).

By hypothesis (\ref{momentum}), we can pick a sequence $T_k\rightarrow\infty$ such that for $k$ big enough

\be T_k^{\frac s 2}\sup_u \P(\mathcal M_{u,\omega}(p,\mathcal X,T_k)>T_k^\alpha)<p_0/12. \ee
In an analogous way we can estimate (\ref{probXii}) and (\ref{probXiii}).  It follows that for all $E\in I$
we have

\be\displaystyle \limsup_{k\rightarrow \infty}\sup_y\mathbb P \left(
\frac{1}{L_k^\theta} <
\Vert\Gamma_{y,L_k}R_{\omega,y,L_k}(E)\chi_{y,L_k/3} \Vert_{L_k}
\right)< p_0. \ee

Since $0<p_0<1$ is arbitrary, we conclude that (\ref{CImsa}) holds
for each $E\in I$.

 \end{proof}

\bigskip

\begin{proof}[Proof of Lemma \ref{deterministiclemma} ] 

Let $\omega\in\set{\omega:\mbox{ } \mathcal M_{u,\omega}(p,\X,T)\leq T^\alpha}$. For a given compact subinterval $I\subset J$, $M>0$ and
$L=L(\epsilon,I)$ as in (\ref{L}),  we set

\[A_{u,\omega,M,I}=\{ E\in I:\mbox{
}\norm{\angles{X-u}^{p/2}R_\omega(E+i\epsilon)\X(\Homega)\chi_u}_2^2\leq
M\epsilon^{-(\alpha+1)} \}.  \]
We have, using  \cite[Lemma 6.3]{GK3}

\begin{align}\label{Acomp} |I\setminus A_{u,\omega,M, I}|& \leq \frac{1}{M\epsilon^{-(\alpha+1)}}\int_\R \norm{\angles{X-u}^{p/2}R_\omega(E+i\epsilon)\X(\Homega)\chi_u}_2^2dE\nonumber\\
& =  \frac{2\pi}{ M T^{\alpha+1}}\int_0^\infty e^{-2t/T} \norm{\angles{X-u}^{p/2}e^{-it\Homega}\X(\Homega)\chi_u}_2^2dt\nonumber\\
& = \frac{\pi}{ M T^\alpha} \mathcal M_{u,\omega}(p,\X,T)\nonumber\\
& \leq \frac \pi { M},
\end{align}
where the last  bound is uniform on $u$ and $\omega$.
\smallskip

Thus, for an $E\in I$ fixed either $E\in A_{u,\omega, M,I}$ in which case we have

\begin{align}\norm{\Gamma_{y,L} R_\omega(E+i\epsilon)\mathcal
X(H_\omega)\chi_{u} }& \leq C_{p,d}L^{-p/2}\norm{\angles{X-u}^{p/2}R_\omega(E+i\epsilon)\X(\Homega)\chi_u}_2\nonumber\\
& \leq C_{p,d}L^{-p/2}M^{1/2}\epsilon^{-(\alpha+1)/2}
 \end{align}
or else, $E\in I\setminus A_{u,\omega, M,I} $, so by \ref{Acomp} there exists $E_{u,\omega}\in A_{u,\omega, M,I}$ such that 

\[ |E-E_{u,\omega}|\leq \frac{\pi}{ M}\]
and therefore, by the resolvent identity and the definition of $A_{u,\omega,M,I}$,

\begin{align}\label{RI} \norm{\Gamma_{y,L} R_\omega(E+i\epsilon)\mathcal
X(H_\omega)\chi_{u} } & \leq 
\norm{\Gamma_{y,L} R_\omega(E_{u,\omega}+i\epsilon)\mathcal
X(H_\omega)\chi_{u} } \nonumber\\ & + |E-E_{u,\omega}| \Vert
R_{\omega}(E+i\epsilon) \Vert \Vert
R_{\omega}(E_{u,\omega}+i\epsilon) \Vert \nonumber\\
&\leq C_{p,d}L^{-p/2}M^{1/2}\epsilon^{-(\alpha+1)/2}  +
\frac{\pi}{M\epsilon^2} 
\end{align}

Now, for $p$ such that $p>p(\theta,\gamma,d,\alpha,s)$ we have 

\be 2(\theta+\gamma+d)<p -6\theta-\gamma-(1+\alpha)(\theta s+d)/s\ee
so if we set
\be M= L^{6\theta+\gamma}, \ee
and recall 
\be\epsilon
^{-(1+\alpha)/2}= C_{p_0,Q_I}L^{(1+\alpha)(\theta s+d)/2s},\ee
we obtain, for $L$ big enough depending on $d,I,p,\alpha,\theta,\gamma,s,p_0, Q_I$,

\begin{align}
C_{p,d}L^{-p/2}M^{1/2}\epsilon^{-(\alpha+1)/2} & = C_{p,d,Q_I,p_0}L^{-(p/2 -(6\theta+\gamma)/2-(1+\alpha)(\theta
s+d)/2s)}\nonumber\\ & <\frac{1}{4L^{(\theta+\gamma+d)}}
\end{align}
and

\be \frac{\pi}{ M\epsilon^2}= C'_{p_0,Q_I}L^{6\theta+2\gamma-2(\theta s +d)/s}<
\frac{1}{4L^{(\theta+\gamma+d)}}.
\ee

Inserting this in (\ref{RI}) proves the
lemma.

\end{proof}

\smallskip


\section{Uniform Wegner estimates for Delone-Anderson type potentials} \label{UWE}

\smallskip

\begin{defn}
A subset $D$ of $\Rd$ is called an \emph{(r,R)-Delone set}
if there exist reals $r$ and $R$ such that for any cubes
$\L_r$, $\L_R$ of sides $r$ and $R$ respectively, we have
$\sharp(D\cap\L_r)\leq1$ and $\sharp(D\cap\L_R)\geq 1$,
where $\sharp$ stands for cardinality.
\end{defn}

\begin{rem} {\it Note that in an $(r,R)$-Delone set there exists a minimal distance
between any two points, $r/2$, and a maximal distance between neighbors, $R/\sqrt{2}$.  Such a set
is said to be \emph{uniformly discrete} and \emph{relatively dense}.  A lattice is a particular case of a Delone set.}
\end{rem}

Take $0 < r< R<\infty$ and consider the operator $H_\omega=H_0+\lambda V_\omega$  with random potential given by

\be \label{ranpot1}
V_\omega(x)=\displaystyle\sum_{\gamma \in D}\omega_\gamma
u(x-\gamma),
\ee
where $D$ is a $(r,R)$-Delone set.  The measurable function $u$, called
\emph{single-site potential}, is such that $\Vert \displaystyle\sum_{\gamma \in D}u(\cdot-\gamma)
\Vert_\infty =1$, it has compact support and satisfies
\be\label{u} u^-\chi_{0, \epsilon_u}\leq u \leq u^+\chi_{0,
\delta_u}, \ee
for some constants $ 0<\epsilon_u\leq \delta_u <\infty $ and $0< u^-
\leq u^+ < \infty $.

\smallskip

Here, $(\omega_\gamma)_{\gamma\in D}$ is a family of
independent random variables, with
probability
distributions $\mu_\gamma$ of bounded and continuous densities $\rho_\gamma$ such that

\be \rho_+ :=\sup_{\gamma\in D} \Vert \rho_\gamma \Vert_\infty <\infty,  \ee

\be 0\in \mbox{supp }\rho_\gamma\subset [-m_0,M_0]  \ee
where $0\leq m_0<\infty$, $0<M_0<\infty$.

Under these assumptions $V_\omega$ is a bounded scalar potential
jointly measurable in both $\omega\in\Omega$ and $x\in\mathbb R^d$, and so the
mapping $\omega\mapsto
H_{\omega}$ is measurable.

Denote by $H_{\lambda,\omega,x,L}$ and $H_{0,x,L}$ the restriction of $H_\omega$ and $H_0$ to the cube $\L_L(x)$ with periodic boundary conditions, respectively (in the particular case of the Landau Hamiltonian, details on the finite volume operator $H_{B,L}$ are stated in Section \ref{Landau}),
with $\lambda$ fixed and $V_{\omega,x,L}$ being the restriction of $V_\omega$  to $\L_L(x)$, defined by

\be\label{restpot}
V_{\omega,x,L}(\cdot)=\displaystyle\sum_{\gamma\in
D\cap\L_{L-\delta_u}(x)}\omega_\gamma u(\cdot-\gamma). \ee
and denote by  $\tilde V_{x,L}$ the potential defined by 
\begin{equation}\label{5}
\tilde V_{x,L}(\cdot) = \displaystyle \sum_{\gamma \in
\tilde\L_{L-\delta_u}(x)}u(\cdot-\gamma).
\end{equation}
 where $\tilde\L_L(x)=D\cap \L_L(x)$.

 We denote by $P_{\lambda,\omega,x,L}$, $P_{0,x,L}$  the spectral projector associated to the finite volume operators $H_{\lambda,\omega,x,L}$, $H_{0,x,L}$, respectively.  In the particular case of the finite volume random Landau Hamiltonian and free Landau Hamiltonian, we write $H_{B,\lambda, \omega, x,L}$ and $H_{B,x,L}$, respectively, and we use the notation $\Pi_{n,x,L}$ for the spectral projector associated to the $n$-th Landau level, and $\Pi_{n,x,L}^\bot$ for its orthogonal projector (see Section \ref{Landaumodel}).
Define $s(\epsilon)=\displaystyle \sup_{\gamma\in D} \sup_{E\in\mathbb R}
\mu_\gamma ([E,E+\epsilon])$.  

We prove several Wegner estimates that we summarize in the following theorem,

\smallskip

\begin{thm}\label{thm}\label{WE1}

\begin{itemize}
\smallskip
\item[i.] For $d=2$, let $H_0$ be the Landau Hamiltonian with constant magnetic field $B>0$ fixed.
For any bounded interval $I\in\mathbb R$ there exist constants $Q_W=Q_W(B,\lambda,R,r,I,u,m_0,M_0)$, $\eta_{B,\lambda,\Delta}\in]0,1]$ and a finite scale $\mathcal L_*(B,\lambda,I,R)$ such that for every compact subinterval $\Delta\subset I$, with $|\Delta|<\eta_{B,\lambda,\Delta}$ and $L>\mathcal L_*$, we have

\be \label{UWE} \sup_{x\in\R^d}\mathbb E \{\mbox{ tr } P_{\lambda,\omega,x,L}(\Delta)\} \leq
Q_{W}\rho_+s(|\Delta|)L^d. \ee

\bigskip
%

\bigskip

\item [ii.] Let $E_0\in\mathbb R\setminus \sigma(H_0)$ for
$H_0=-\Delta+V_0$, where $V_0$ is $\mathbb Z^d $-periodic.  For any bounded interval $I\subset \mathbb R\setminus \sigma(H_0)$ there exist a constant $Q_W=Q_W(\lambda,R,r,I,u)$ and a finite scale $\mathcal L_*(R)$ such that for every compact subinterval $\Delta\subset I$, (\ref{UWE}) holds.

\bigskip

\item[iii.]Assume the IDS of $H_0$ is H\"older
continuous with exponent $\delta>0$ in some open interval $I$ and no
further assumption on $s(\epsilon)$. Then there exists a constant
$Q_W'=Q_W'(B,\lambda,I,u, R,r,d)>0$ such that for all compact subintervals $\Delta\subset
I$  with $|\Delta|$ small enough, and $0<\gamma<1$,

\be \mathbb E \{\tr P_{\lambda,\omega,x,L}(\Delta)\} \leq
Q_W'\max\{|\Delta|^{\delta\gamma},|\Delta|^{-2\gamma}s(|\Delta|)
 \}L^d.
\ee
In particular, if $s(\epsilon)\leq C\epsilon^\zeta$, for some
$\zeta\in [0,1]$, then

\be \mathbb E \{\tr P_{\lambda,\omega,x,L}(\Delta)\} \leq
Q_W'|\Delta|^{\frac{\zeta\delta}{\delta+2}}L^d. \ee

\smallskip

\smallskip

\end{itemize}

\end{thm}

\bigskip

Since the results are uniform in $x$, we state them for $x=0$, $\lambda$ fixed and for simplicity we
omit these subscripts from the notation.

For the proof we follow \cite{CHK2}, based on \cite{CHK}, plus \cite{GKS} in the case of the Landau Hamiltonian. 
In all cases we need to estimate $\mathbb E\{\tr P_{\omega,L}(\Delta)\}$.  We
decompose it with respect to the free spectral projector of an interval $\tilde \Delta$, such that $\Delta\subset \tilde \Delta$ and $d_{\Delta}=dist(\Delta, \tilde\Delta^c)>0$, that is

\be\label{decomp} \tr P_{\omega,L}(\Delta)=\tr P_{\omega,L}(\Delta)P_{0,L}(\tilde \Delta)+\tr
P_{\omega,L}(\Delta)P_{0,L}(\tilde \Delta^c). \ee
The key step in estimating the first term of the r.h.s  is to prove a positivity estimate as in  \cite[Theorem 2.1]{CHK2}.  In order to obtain this estimate in the case of the Landau Hamiltonian, we need some preliminary lemmas. 

\begin{lem}\label{PEthm}
Using the notations above, there exists a positive finite constant 
 $C_n(B,u,R)$, so that

\begin{equation}\label{PE}
\Pi_{n,L}\tilde V_{x,L} \Pi_{n,L} \geq C_n(B,u,R)\Pi_{n,L}.
\end{equation}
\end{lem}

\begin{proof}
From \cite{CHKR} we have that for $n \in \mathbb N$, $\tilde R >0$, for each
$0<\epsilon < \tilde R$, $\kappa>1$ and $\eta>0$ there exists a constant
$C_0=C_{0,n,\epsilon,\tilde R,\eta}>0$ such that

\be \label{chi0} \Pi_n\chi_{0,\epsilon}\Pi_n \geq
C_0(\Pi_n\chi_{0,\tilde R}\Pi_n - \eta \Pi_n\chi_{0,\kappa \tilde R}\Pi_n ).
 \ee

Because of the invariance of $H_B$ under the magnetic
translations (\ref{MT}) we have that the projections $\Pi_n$ commute
with these unitary operators, which in turn gives, for an arbitrary
$x\in \mathbb R^2$,

\bea\label{chix}
 U_x\Pi_n\chi_{0,\epsilon}\Pi_n U_x^* \geq
C_0 U_x(\Pi_n\chi_{0,\tilde R}\Pi_n - \eta \Pi_n\chi_{0,\kappa \tilde R}\Pi_n
)U_x^*
\\
\Pi_n U_x \chi_{0,\epsilon}U_x^*\Pi_n \geq C_0
(\Pi_nU_x\chi_{0,\tilde R}U_x^*\Pi_n - \eta \Pi_nU_x\chi_{0,\kappa
\tilde R}U_x^*\Pi_n )\\
\Pi_n \chi_{x,\epsilon}\Pi_n \geq C_0 (\Pi_n\chi_{x,\tilde R}\Pi_n - \eta
\Pi_n\chi_{x,\kappa \tilde R}\Pi_n ),
 \eea
since conjugation by unitary operators is a positivity preserving
operation.

Now, we recall \cite[Lemma 5.3]{GKS} (which is independent of
$V$ and, therefore, $D$).

\begin{lem}\label{lemmaPE} Fix $B>0$, $n\in\mathbb N$, $\tilde R >0$, $0<\epsilon< \tilde R$ and
$\eta>0$.
 If $\kappa>1$ and $L\in\mathbb N_B$ (defined as in (\ref{LforLandau}))  are such that $L>2(L_B+\kappa
 \tilde R)$ then for all $\tilde x\in\L_L(x)$, we have

 \be
\Pi_{n,L}\hat\chi_{\tilde x,\epsilon}\Pi_{n,L}\geq
C_0\Pi_{n,L}(\hat\chi_{\tilde x, \tilde R}-\eta\hat\chi_{\tilde x,\kappa
\tilde R})\Pi_{n,L} +\Pi_{n,L}\mathcal E_{n,\tilde x,L}\Pi_{n,L},
 \ee
where $C_0=C_{0;n,B,\epsilon,\tilde  R,\eta}>0$ is a constant as before and the error
operator $\mathcal E_{n,\tilde x,L}$ satisfies

\be\label{EO} \Vert \mathcal E_{n,\tilde x,L} \Vert \leq
C_{n,B,\epsilon,R,\eta}e^{-m_{n,B}L},
 \ee
for some positivie constant $m_{n,B}$.
\end{lem}

\smallskip

Now, by (\ref{u}) we have

\be \tilde V_{x,L}(\cdot) = \displaystyle \sum_{\gamma \in
\tilde\L_{L-\delta_u}(x)}u(\cdot-\gamma) \geq
u^-\displaystyle\sum_{\gamma \in \tilde\L_{L-\delta_u}(x)}
\hat\chi_{\gamma,\epsilon_u}.
 \ee

We fix $\tilde R>2 R+\delta_u$, in which case \be\label{chi1}
\displaystyle\sum_{\gamma \in \tilde\L_{L-\delta_u}(x)}
\hat\chi_{\gamma,\tilde R} \geq \chi_{x,L}. \ee

Now fix $\kappa >1$ and pick $\eta>0$ such that

\be\label{chi2}
 \eta\displaystyle\sum_{\gamma \in
\tilde\L_{L-\delta_u}(x)} \hat\chi_{\gamma,\kappa \tilde R} \leq
\frac{1}{2}\chi_{x,L}. \ee

It follows from Lemma \ref{lemmaPE}, (\ref{chi1}) and (\ref{chi2})
that

\bea
\Pi_{n,L}\tilde V_{x,L} \Pi_{n,L}\geq u^- C_0
 \displaystyle\sum_{\gamma\in\tilde\L_{L-\delta_u}(x)}\Pi_{n,L}(\hat\chi_{
\gamma,\tilde R}-\eta\hat\chi_{\gamma,\kappa \tilde R})\Pi_{n,L}
+\Pi_{n,L}\mathcal E_{n,L}\Pi_{n,L}
\\
\geq \frac{u^- C_0}{2} \Pi_{n,L} + \Pi_{n,L}\mathcal
E_{n,L}\Pi_{n,L}
\\
\geq C_1\Pi_{n,L}, \eea
for $L\geq L^*$ for some
$L^*=L_{n,B,\epsilon,R,\kappa,\eta}^*<\infty$ and
$C_1=\frac{u^-C_0}{4}$, since the error operator
\[ \Pi_{n,L}\mathcal
E_{n,L}\Pi_{n,L}=\Pi_{n,L}\displaystyle\sum_{\gamma\in\tilde\L_{L-\delta_u}
(x)}\mathcal
E_{n,\gamma,L}\Pi_{n,L}   \]
by (\ref{EO}), satisfies
\[ \Vert \mathcal E_{n,L}\Vert\leq L^2C_{n,B,\epsilon,R,\eta}e^{-m_{n,B}L}.   
\]

\end{proof}

Finally we recall,

\begin{lem}\cite[Lemma 2.1]{CHK2}\label{lemaux}
Suppose that T is a trace class operator independent of $\omega$ and u, the
single site potential (\ref{u}).  We then have

\be \mathbb E \{\tr P_{\omega,L}(\Delta)u_iTu_j\} \leq 8 s(|\Delta|)\Vert
u_iTu_j \Vert_1. \ee
where we use the notation $u_i=u(x-i)$, $i\in\R^2$.

\end{lem}
\bigskip

\begin{proof}[Proof of Theorem \ref{thm}]

To prove $(i)$, using the preliminary lemmas we can follow the proof in \cite[Theorem 4.3]{CHK2}.  Notice that the
spatial homogeneity of the Delone set in the sense that points do not
accumulate neither are too far away, so the sums over indexes of elements of $D$
preserves the properties of the sums over indexes of elements of the lattice
$\mathbb Z^2$ as the original proofs.
\smallskip

Recall that we need to estimate $\mathbb E\{\tr P_{\omega,L}(\Delta)\}$ as in (\ref{decomp}), that is, for an arbitrary
$E_0\in\mathbb R$ , with  $\Delta$ and $\tilde\Delta$  closed bounded intervals centered on $E_0$ such that $\Delta\subset \tilde\Delta$,  $ |\Delta|<1$, $d_\Delta
>0$, we need to estimate

\be\label{decomp1} \tr P_{\omega,L}(\Delta)=\tr P_{\omega,L}(\Delta)\Pi_{n,L}+\tr
P_{\omega,L}(\Delta)\Pi_{n,L}^\bot. \ee

\bigskip

\emph{\bf{a.} Estimate on } $\mathbb E \{\tr
P_{\omega,L}(\Delta)\Pi_{n,L}^\bot\}.$\newline

The analysis in \cite[Eq. 2.6 - 2.10]{CHK2} for the $n$-th Landau band remains valid
taking, for the constants defined therein, $M=1$ and the operator K defined by

\be K\equiv \left(
\frac{H_{B,L}+1}{H_{B,L}-E_m}\right)^2,\hspace{1cm}\Vert K\Vert\leq
K_n\equiv \left( 1+ \frac{1+\Delta_+}{d_n}\right)^2,\ee
where $E_m$ is an eigenvalue of $H_{B,\lambda,\omega,L}$, $d_n\equiv
\min\{\dist(I, B_{n-1}), \dist(I, B_{n+1}\}$ and
$\Delta=[\Delta_{-},\Delta_+]$.

\bigskip

Then we can obtain the analog of \cite[Eq. 4.4]{CHK2},

\be \tr P_{\omega,L}(\Delta)\Pi_{n,L}^\bot \leq K_n \lambda^2
\max\{m_0,M_0\}^2\displaystyle\sum_{i,j\in\tilde\L}|\tr\mbox{ }
u_jP_{\omega,L}(\Delta)u_iK_{ij}|,
 \ee
where $K_{ij}\equiv \chi_i(H_{B,L}+1)^{-2}\chi_j$, for $\chi\geq0$ a
smooth function of compact support slightly larger than the support
of $u$ such that $\chi u =u$.  Note that due to the spatial homogeneity of 
$D$ and the fact that $\mbox{supp }u$ is contained in a cube of side $r$, the translated supports of $u$ do not overlap.

\bigskip

Now, denote by $\tilde\L_0=\{i,j\in\tilde\L /
\chi_i\chi_j=0 \}$ and by
$\tilde\L_0^c=\{i,j\in\tilde\L/\chi_i\chi_j\neq 0\}$. For
$i,j\in\tilde\L_0$, the operator $K_{ij}$ is trace class  \cite[Lemma
2.2]{BGKS},  \cite[Lemma 5.1]{CHK2} and it satisfies the
Combes-Thomas estimate,

\be \Vert K_{ij}\Vert_1=\Vert \chi_i(H_{B,L}+1)^{-2}\chi_j
\Vert_1\leq C_0'e^{-\tilde C_0\Vert i-j\Vert},\ee
where $C_0'$ and $\tilde C_0$ are positive constants.  So we can use
Lemma \ref{lemaux} to obtain

\bea \mathbb E\{|\displaystyle\sum_{i,j\in\tilde\L_0}\tr \mbox{
} u_jP_{\omega,L}(\Delta)u_iK_{ij}| \} \leq \mathbb E \{
\displaystyle\sum_{i,j\in\tilde\L_0}|\tr\mbox{ }
u_jP_{\omega,L}(\Delta)u_iK_{ij}|\}\\
\leq C_08s(|\Delta|)\displaystyle\sum_{i,j\in\tilde\L_0}
e^{-\tilde C_0 \Vert i-j\Vert}\\
\leq C_1 s(|\Delta|)|\L|.  \eea
where $C_1$ also depends on $r$, since $\sharp (\tilde \L_L)\leq C_{r,d} L^d$ for $L>R$, see Eq. (\ref{DeloneDense}).

On the other hand, for $i,j\in\tilde\L_0^c$, $K_{ij}$ is also
trace class \cite[Lemma 2.2]{BGKS} so we can apply Lemma \ref{lemaux} again, obtaining

\be \mathbb E\{\tr P_{\omega,L}(\Delta)\Pi_{n,L}^\bot\}\leq \ C_2 s(|\Delta|)|\L|, \ee
where $C_2>0$ depends on $u$, $I$, $\lambda$, $r$ and $M=\max\{m_0,M_0\}$ .

\bigskip

\emph{\bf b. Estimate on} $\mathbb E \{\tr
P_{\omega,L}(\Delta)\Pi_{n,L}\}.$\newline

We use the spectral projector $\Pi_{n,L}$ in order to control the
trace.  Here the key ingredient is the positivity estimate
(\ref{PE}) and the fact that, under our
 hypotheses on $u$, there exists a finite constant $C_u$, depending on
 $u$ only, such that
 \[ 0<\tilde V_L^2\leq C_u\tilde V_L. \]

Now,

\bea \tr P_{\omega,L}(\Delta)\Pi_{n,L}\leq \frac{1}{C_n(B,u,R)}\tr
P_{\omega,L}(\Delta)\Pi_{n,L}\tilde V_L\Pi_{n,L} \\
\leq \frac{1}{C_n(B,u,R)}\left\{\tr P_{\omega,L}(\Delta)\tilde
V_L\Pi_{n,L}-\tr P_{\omega,L}(\Delta)\Pi_{n,L}^\bot\tilde
V_L\Pi_{n,L} \right\}.\eea

Then we can proceed as in parts (2) and (3) of the proof of \cite[Theorem 4.3]{CHK2},  and we
finally arrive to the desired result,

\be \mathbb E\{\tr P_{\omega,L}(\Delta)\}\leq Q_Ws(|\Delta|)|\L|. \ee
where the constant $Q_W>0$ depends on $B,u, R,r,I,\lambda$ and $M$.

\bigskip

As for $(ii)$, note that in this case $\tr P_{0,L}(\tilde\Delta)=0$ if
$\tilde\Delta\subset\mathbb R\setminus\sigma(H_0)$, so we only need to estimate the second term in the r.h.s. of (\ref{decomp}), where we do not need the positivity estimate (\ref{PE}) for $P_{0,L}$.  The proof mimics $(i)$-$a$.

\bigskip

In case $(iii)$ we can estimate the first term in the r.h.s. of (\ref{decomp}) without using the analog of (\ref{PE}) for $P_{0,L}$. Instead, the H\"older continuity of
the IDS of the non perturbed operator implies that there exists a constant $C>0$ such that

\[ \tr P_{0,L}(\tilde\Delta)\leq C|\tilde\Delta|^\delta|\L|, \]
and so, for $0<\gamma<1$
\be \tr P_{\omega,L}(\Delta)P_{0,L}(\tilde\Delta)\leq
C|\Delta|^{\gamma\delta}|\L|. \ee

 Since, as in the previous case (writing explicitly the dependence on $d_\Delta$) we have

 \[ \mathbb E \{\tr P_{\omega,L}(\Delta)P_{0,L}(\tilde\Delta^c)\}\leq
\frac{Q_W'}{d_{\Delta}^2}
s(|\Delta|)|\L|,  \]
by taking $d_\Delta=|\Delta|^\gamma$ we obtain the desired result.
Furthermore, if $s(\epsilon)$ is $\zeta$-H\"older continuous, we
get, taking $\gamma$ such that $\gamma\delta=\zeta-2\gamma$,

\bea \mathbb E \{\tr P_{\omega,L}(\Delta)\}& \leq Q_W'\max\{
|\Delta|^{\gamma\delta}, |\Delta|^{\zeta-2\gamma} \}L^2 \\ & \leq
Q_W'|\Delta|^{\frac{\zeta\delta}{\delta+2}}L^2, \eea
where $Q_W'$ depends on $u,I,\lambda$, $R$, $r$ and $M$.

\bigskip

\end{proof}

\bigskip


\section{Applications to non ergodic random Landau operators}\label{Landau}

\subsection{The model} \label{Landaumodel}

We consider the case where the free Hamiltonian in (\ref{ranop}) is $H_B$, the Landau Hamiltonian,  and the random potential represents impurities placed in a Delone set (for the case 
 $H_0=-\Delta$  see \cite{G}).  We aim to prove for this model the existence of complementary regions of dynamical localization and delocalization in the spectrum and therefore, the existence of a dynamical transition energy.  By doing this we extend known results for ergodic random Landau Hamiltonians \cite{CH,CH2,GKS,GKS2} to non-ergodic ones.

\smallskip
Let $H_B$ be the unperturbed Landau Hamiltonian on $\Lp{\mathbb
R^2}$
\begin{equation}\label{H_B}
H_B = (-i\nabla-{\bf A})^2 \quad \hbox{ with } {\bf A}
=\frac{B}{2}(x_2,-x_1),
\end{equation}
where {\bf A} is the vector potential and $B$ is the strength of
the magnetic field.

 \bigskip
 The spectrum of $H_B$ is pure point and consists of a sequence of
infinitely degenerate eigenvalues, the Landau levels $\{
B_n=(2n+1)|B|;\mbox{ } n=0,1,... \}$, with associated orthogonal projection
operators $\Pi_n$.  As the spectrum is independent of the sign of B,
we will always assume $B>0$.

\bigskip

We define the magnetic translations $U_a$ for $a\in \mathbb R^2$ and
$\varphi\in \mathcal C^\infty_0(\mathbb R^2)$, by

\begin{equation}\label{MT}
U_a\varphi (x)= e^{-i\frac{B}{2}(x_2a_1-x_1a_2)}\varphi(x-a),
\end{equation}
\\
obtaining a projective unitary representation of $\mathbb R^2$ on
$\Lp{\mathbb R^2}$:
\\
\begin{equation}\label{4}
U_a U_b=
e^{i\frac{B}{2}(a_2b_1-a_1b_2)}U_{a+b}=e^{iB(a_2b_1-a_1b_2)}U_b U_a,
\quad a,b\in \mathbb R^2.
\end{equation}
\\
We then have $U_a H_B U_a^*=H_B$ for all $a\in \mathbb R^2$.

\bigskip

We consider the perturbed family of Landau Hamiltonians
given by

\be\label{randop}
 H_{B,\lambda,\omega}= H_B +\lambda V_\omega \hspace{1cm}\mbox{on}
\hspace{0.3cm} \Lp{\mathbb R^2}, \ee
where, as before, $\lambda$ is the disorder parameter which we consider fix and $V_\omega$ is the Delone-Anderson type potential given by (\ref{ranpot1})-(5.4) with the additional conditions:

(\emph{uc.}) $\delta_u < \tilde r/10$, i.e. $u$ has compact support
contained in $B(0,\tilde r/10)$. This implies that for $i,j\in D$
with $i\neq j$, $\mbox{supp } u_i\cap \mbox{supp }u_j=\emptyset$, where we use the
notation $u_i=u(\cdot-i)$ for $i\in\mathbb R^2$.\label{uc}

(\emph{u0.}) $\norm{u}_\infty=1$ and $u(0)=1$\label{u0}.

\smallskip

We denote the spectrum of this operator by
$\sigma_{B,\lambda,\omega}$.  By perturbation theory \cite[Theorem V.4.10]{K} we know that for each $\omega\in\Omega$

\[ \sigma_{B,\lambda,\omega} \subset \displaystyle\bigcup_{n=0}^\infty \mathcal
B_n(B,\lambda),  \]
where $\mathcal B_n(B,\lambda)=[B_n-\lambda m_o,B_n+\lambda M_0]$ is
called the \emph{n-th Landau band}.  Moreover, by a Borel-Cantelli argument, for almost every $\omega\in\Omega$,
\be \sigma_{B} \subset \sigma_{B,\lambda,\omega} \ee 
where $\sigma_B$ is the spectrum of the free Landau operator.  We also show that there exists almost surely spectrum near the band edges so our results are not empty (see  Section \ref{spectrum})

For $B$ fixed
$\lambda$ is small enough such that

\be\label{DBC}  \lambda(m_0+M_0)<2B,  \ee
i.e., the Landau bands $\mathcal B_n(B,\lambda)$ are disjoint and
hence the open intervals

\be \mathcal G_n(B,\lambda)=]B_n+\lambda M_0, B_{n+1}-\lambda m_0[,
\hspace{1cm} n=0,1,2,..., \ee
are nonempty spectral gaps for $H_{B,\lambda,\omega}$.

\bigskip

We now define finite volume operators following \cite{GKS}.  For $B>0$, we set
 \begin{equation}\label{LforLandau}
 K_B=\min\left\{ k\in\mathbb N : k\geq
\sqrt{\frac{B}{4\pi}} \right\} \quad \mbox{and} \quad
L_B=K_B\sqrt{\frac{B}{4\pi}}.
\end{equation}
We denote $\mathbb N_B=L_B\mathbb N$, $\tilde{\mathbb N_B} =\mathbb
N_B \cup \{\infty\}$ and $\mathbb Z_B^2= L_B \mathbb Z^2$.

\bigskip

We consider squares $\L_L(x)$ with $L\in\mathbb N_B$ and
$x\in\mathbb R^2$, and identify them with the torii $\mathbb
T_{L,x}:=\mathbb R^2/(L\mathbb Z^2+x)$.
We denote by $\chi_{x,L}$ the characteristic function of the cube
$\L_L(x)$ and for $\tilde x\in \L_L(x)$ and $r<L$ we
denote by $\hat\L_r(\tilde x)$ and $\hat\chi_{\tilde x,r}$ the
cube and characteristic function in $\mathbb T_{L,x}$.

\bigskip

For the first order differential
operator ${\bf D}_B=(-i\nabla-{\bf A})$ restricted to $\mathcal
C^\infty_c(\L_L(x))$ we take its closed, densely defined
extension ${\bf D}_{B,x,L}$ from $\Lp{\L_L(x)}$ to
$\Lp{\L_L(x);\C^2}$, with periodic boundary conditions
and then set $H_{B,x,L}={\bf D}_{B,x,L}^*{\bf D}_{B,x,L}$.  

We are left with the operator $H_{\omega,B,x,L}$ acting on
$\Lp{\L_L(x)}$ defined by

\be H_{B,\lambda,\omega,x,L}=H_{B,x,L}+ \lambda V_{\omega,x,L}.  \ee 
where $V_{\omega,x,L}$ is defined as in \ref{ranpot1}

We write
$R_L(z)=(H_{B,\lambda,\omega,x,L}-z)^{-1}$ for the resolvent operator of
$H_{B,\lambda,\omega,x,L}$.

\smallskip

Since $H_{B,x,L}$ has a compact resolvent, its spectrum consists in
the Landau Levels but now with finite multiplicity.  We denote by
$\Pi_{n,L}$ the orthogonal projection associated to the $n$-th Landau level and define
$P_{B,\lambda,\omega,x,L}(J)=\chi_J(H_{B,\lambda,\omega,x,L})$ for
$J\subset\mathbb R$ a Borel set.

\bigskip

This operator satisfies the compatibility conditions \cite[Eq. 
4.2]{GKS}: If $\varphi\in\mathcal D({\bf D}_{B,x,L})$ with
supp $\varphi\subset\L_{L-\delta_u}(x)$, then $\mathcal
I_{x,L}\varphi\in\mathcal D({\bf D}_B)$ and

\be\label{cc}
\begin{array}{cc}
\mathcal I_{x,L}{\bf D}_{B,x,L}\varphi= {\bf
D}_B\mathcal I_{x,L}\varphi, \\
\mathcal
I_{x,L}\chi_{x,L-\delta_u}V_{\omega,x,L}=\chi_{x,L-\delta_u}V_\omega,
\end{array}
\ee
where $\mathcal I_{x,L}:\Lp{\L_L(x)}\rightarrow \Lp{\mathbb
R^2}$ is the canonical injection

\begin{displaymath}
    \mathcal I_{x,L}\varphi(y)=\left\{ \begin{array}{ll}
\varphi(y) & \textrm{if $y
\in\L_L(x)$}\\
0 & \textrm{otherwise}.\\
\end{array} \right.
\end{displaymath}

From this we have
\[ \mathcal I_{x,L}H_{B,\lambda,\omega,x,L}\varphi=H_{B,\lambda,\omega}\mathcal
I_{x,L}\varphi, \]
that is, the finite volume operators $H_{B,\lambda,\omega,x,L}$
agree with $H_{B,\lambda,\omega}$ inside the square $\L_L(x)$.

\smallskip

However, $H_{B,\lambda,\omega,x,L}$ does not satisfy the covariance
condition (\ref{nonerg}) so we have \emph{a priori}

\[ H_{B,\lambda,\omega,x,L}\neq U_xH_{B,\lambda,\tau_{-x}(\omega),0,L}U_x^*, \]
where $U_x$ is the magnetic translation (\ref{MT}) seen as a unitary
map from $\Lp{\L_L(0)}$ to $\Lp{\L_L(x)}$ and $\tau_{x}$
is the translation defined as
$\tau_{x}(\omega_\gamma)=\omega_{\gamma-x}$ for $x\in\mathbb R^2$.

\bigskip

\subsection{Dynamical localization in Landau bands}

 In this section we prove

\begin{thm}\label{LocLandau}
Let $H_\omega$ be as before.  For any $n=0,1,2,...$ there exist
finite positive constants ${\bf{B}}(n)$ and $K_n(\lambda)$ depending only on
$n$, $M$, $u$ and $\rho$ such that for all $B\geq{\bf{B}}(n)$ we can
perform MSA in the intervals \be \Sigma_{B,n,\lambda,\omega}=
\sigma_{B,\lambda,\omega}\cap\{ E\in\mathcal B_n :\mbox{ } |E-B_n|\geq
K_n(\lambda)\frac{\log B}{B}\}, \ee

We have strong HS-dynamical localization at energy levels up to a distance
$K_n(\lambda)\frac{\log B}{B}$ from the Landau levels for large $B$.

\end{thm}

For the proof we need to verify the conditions to start the modified Multiscale Analysis, Theorem \ref{Bootstrap}.
As mentioned in the proof of Theorem \ref{Bootstrap}, this model satisfies properties (IAD), (R), (EDI), (SLI) and
(UNE).  What is left to prove is the existence of a suitable length scale $L_0$ that satisfies 
 (\ref{ILSE}) and (UWE).  The latter comes from the following
improvement in the Wegner estimate of the previous section and it follows \cite[Theorem 3.1]{CH}.

\bigskip

\begin{thm}\label{WE3} There exists $\tilde B>0$ and a constant $Q_n=\tilde
Q_{n,\lambda,u}\Vert \rho \Vert_\infty$ such
that for all $B>\tilde B$ and for any closed interval
$\Delta\subset\mathcal B_n\setminus\sigma(H_B)$

\be\label{WEL1} \mathbb E\{ \tr P_{B,\lambda,\omega,x,L}(\Delta)\}\leq Q_n
\frac{B}{2(dist(\Delta,B_n))^2}|\Delta|L^2.
\ee 

In particular, for $E_0\notin\sigma(H_B)$ and all
$0<\epsilon<|E_0-B_n|$,

\be\label{WEL2}  \mathbb P\{\dist(\sigma(H_{B,\lambda,\omega,x,L}),
E_0)\leq \epsilon\}\leq Q_n \frac{B}{(|E_0-B_n|-\epsilon)^2}\epsilon
L^2.
\ee

\end{thm}

\smallskip

\begin{proof}
Without loss of generality we work within the first Landau band
$\mathcal B_0$, containing the Landau level $B_0$. Set $M=\Vert
V_\omega\Vert_\infty=\max\{m_0,M_0\}$. Let $\Delta$ be an interval such that
$\Delta\subset\mathcal B_0\setminus \{B_0\}$ and $\inf \Delta >B$, so $\mbox{dist }(\Delta, B_0)>0$ .

Following the same arguments in \cite[Eq. 3.4 - 3.11]{CH}, we get

\be\label{sumaij}\mathbb E\{ \tr P_L(\Delta)\} < \dist(\Delta,B_0)^{-2}M^2\Vert \rho
\Vert_\infty |\Delta| \displaystyle \sum_{i,j\in D} \Vert
\Pi_{0,L}^{ij} \Vert_1, \ee
where $P_L(\Delta)$ stands for $P_{B,\lambda,\omega,x,L}(\Delta)$
and we use the notation $A^{ij}=u_i^{1/2}A u_j^{1/2}$ for any
bounded operator $A$.

\bigskip

To evaluate the sum we consider separately the indices $i,j$ for
which $\Vert i-j\Vert< 4\delta_u$ and those for which $\Vert
i-j\Vert\geq 4\delta_u$, with $\delta_u$ as in (\ref{u}).


Let $\chi_{ij}$ be the characteristic function of supp$(u_i+u_j)$.
Again, as in Thm \ref{WE1}, the translated supports of $u$ behave in
a similar way as in the lattice.  Then we follow the same arguments
therein and obtain, using  \cite[Lemma 2.1]{CH},

\be \label{ij<delta} \displaystyle \sum_{|i-j|<4\delta_u} \Vert
\Pi_{0,L}^{ij} \Vert_1 \leq \Vert u \Vert_\infty^2 \displaystyle
\sum_{|i-j|<4\delta_u} \Vert \chi_{ij}\Pi_{0,L}\chi_{ij} \Vert_1
\leq C_0 B|\L| |\mbox{supp }u|, \ee
where the constant $C_0$ actually depends on the index $n$ of the
Landau level, which in this case is $0$.

\bigskip

Define $\chi_{ij}^+$ to be the characteristic function of the set
$\{ x\in\mathbb R^2 : \Vert x-i\Vert <\Vert x-j\Vert \}$ and denote
$\chi_{ij}^-=1-\chi_{ij}^+$.  Then we obtain

\[\Vert \Pi_{0,L}^{ij}\Vert_1 \leq \Vert
u_j^{1/2}\Pi_{0,L}\chi_{ij}^+\Vert_2
\Vert\chi_{ij}^+\Pi_{0,L}u_i^{1/2}\Vert_2+\Vert
u_j^{1/2}\Pi_{0,L}\chi_{ij}^-\Vert_2
\Vert\chi_{ij}^-\Pi_{0,L}u_i^{1/2}\Vert_2.  \]

\bigskip

Now, if $|i-j|\geq 4\delta_u$, condition (\ref{u}) implies that

\[ \dist(\mbox{supp }\chi_{ij}^+,\mbox{supp }u_j)\geq \frac{\Vert
i-j\Vert}{2}-\delta_u \geq k\Vert i-j\Vert  \]
for some $k>0$.  Similarly for $\dist(\mbox{supp
}\chi_{ij}^-,\mbox{supp }u_i)$. We then obtain

\be \label{ij>delta} \displaystyle \sum_{|i-j|\geq 4\delta_u} \Vert
\Pi_{0,L}^{ij} \Vert_1 \leq C_1 |\mbox{supp }u| |\L|. \ee

Combining (\ref{sumaij}), (\ref{ij<delta}) and (\ref{ij>delta}) we
obtain

\[ \mathbb E \{ \tr P_{L}(\Delta)\}\leq Q_0(\dist(\Delta,B_0))^{-2}\Vert \rho
\Vert_\infty \epsilon B |\L|, \]
where the constant $Q_0$ depends on $\lambda$, $M$,$\Vert
u \Vert_\infty$ and supp $u$.
Taking $\Delta=[E_0-\epsilon,E+\epsilon]$ for small $\epsilon>0$ and applying Chebyshev's inequality we obtain (\ref{WEL2}).

\end{proof}

As for the initial length scale estimate (\ref{ILSE}) to start the multiscale analysis,  we need to verify that for some
$L_0\in6\mathbb N$ sufficiently large (as specified in \cite{GK2}),
given $\theta>0$, $E\in\mathbb R\setminus\sigma(H_{B,L})$,

\be\label{suitable} \mathbb P\left\{\Vert \Gamma_{x,L_0}
R_{B,\omega,x,L_0}(E)\chi_{x,L_0/3} \Vert\leq
\frac{1}{L_0^\theta}\right\}>1-\frac{1}{L_0^p}, \ee
for a suitable choice of $p$, where
$\Gamma_{x,L}=\chi_{\bar\L_{L-1}(x)\setminus\L_{L-3}(x)}$.

\bigskip

To do so we follow the approach \cite{CH} to obtain estimates that we will later state as in
\cite{GK2}.  We need to show that in the annular region between a box of side $L/3$ and $L$, there exists a closed, connected ribbon where the potential $V$ satisfies the condition $|V(x)+B_n-E|>a>0$, for $ E\neq B_n$  with a good probability (\cite[Eq. 4.2]{CH}).  To prove this, Combes and Hislop used bond percolation theory, defining occupied bonds of the lattice as those bonds where the potential satisfies this property.  However, in our case there is no need to use percolation theory since this fact is assured by the assumption (\ref{uc}) on the single-site potential.  More precisely, we will show that there exist ribbons where the potential is zero almost surely.

\bigskip


Let us consider the {\it Voronoi diagram} associated to $D$ \cite{OBSC}.
Since $\tilde \L_{L}=D\cap\L_L$ is a discrete bounded set, we can
write $\tilde \L_{L}=\{p_1,..., p_n\}$, $n\in\mathbb N$.  For each site
$p_i$ we consider its {\it Voronoi cell}, defined as

\[ \mathcal V (p_i)=\{ x\in \mathbb R^2:\Vert x-p_i\Vert\leq\Vert x-p_j \Vert,
j\neq i, 1\leq j\leq n \},\]
i.e., the set of points that are closer to $p_i$ than to any other site in
$\tilde\L_L$. The Voronoi diagram associated to $\tilde\L_L$, denoted
by $\mathcal Vor(\tilde\L_L)$ is a subdivision of $\L_L$ into Voronoi
cells,
\[ \mathcal Vor(\tilde\L_{L})=\displaystyle\bigcup_{1\leq i\leq n} \mathcal
V(p_i).  \]

The edges and vertices of $\mathcal Vor(\tilde\L_L)$ are polygonal
connected lines with the property that the minimal and maximal distances from
any site $p_i$ to an edge or vertex are $r/4$ and $R/2\sqrt{2}$, respectively.

Now, take a covering of $\L_{L/3}$ by a finite collection of Voronoi cells,
$\mathcal V_\L$,
which is a convex polygon.  Its perimeter is a polygonal line $\mathcal C$ that
encloses $\L_{L/3}$ such that $\mathcal C\cap D=\emptyset$. Taking $L$ big
enough with respect to $R$ we have $\mathcal C\subset
\L_{L-3}\setminus\L_{L/3}$.
Moreover, assumption \emph{(uc)} implies that we can always find a
ribbon $\mathcal R$ associated to $\mathcal C$, i.e., a set
\[ \mathcal R =\{ x\in\mathbb R^2 :\mbox{ }dist(x,\mathcal C)< \frac{\tilde
r}{4} - \frac{\tilde r}{10}  \}, \]
such that $V(x)=0$ for all $x\in\mathcal R$ (see Fig. \ref{fig:Voronoi})

\begin{figure}[h]
\centering
\vspace{1mm}
\scalebox{0.5}{\includegraphics{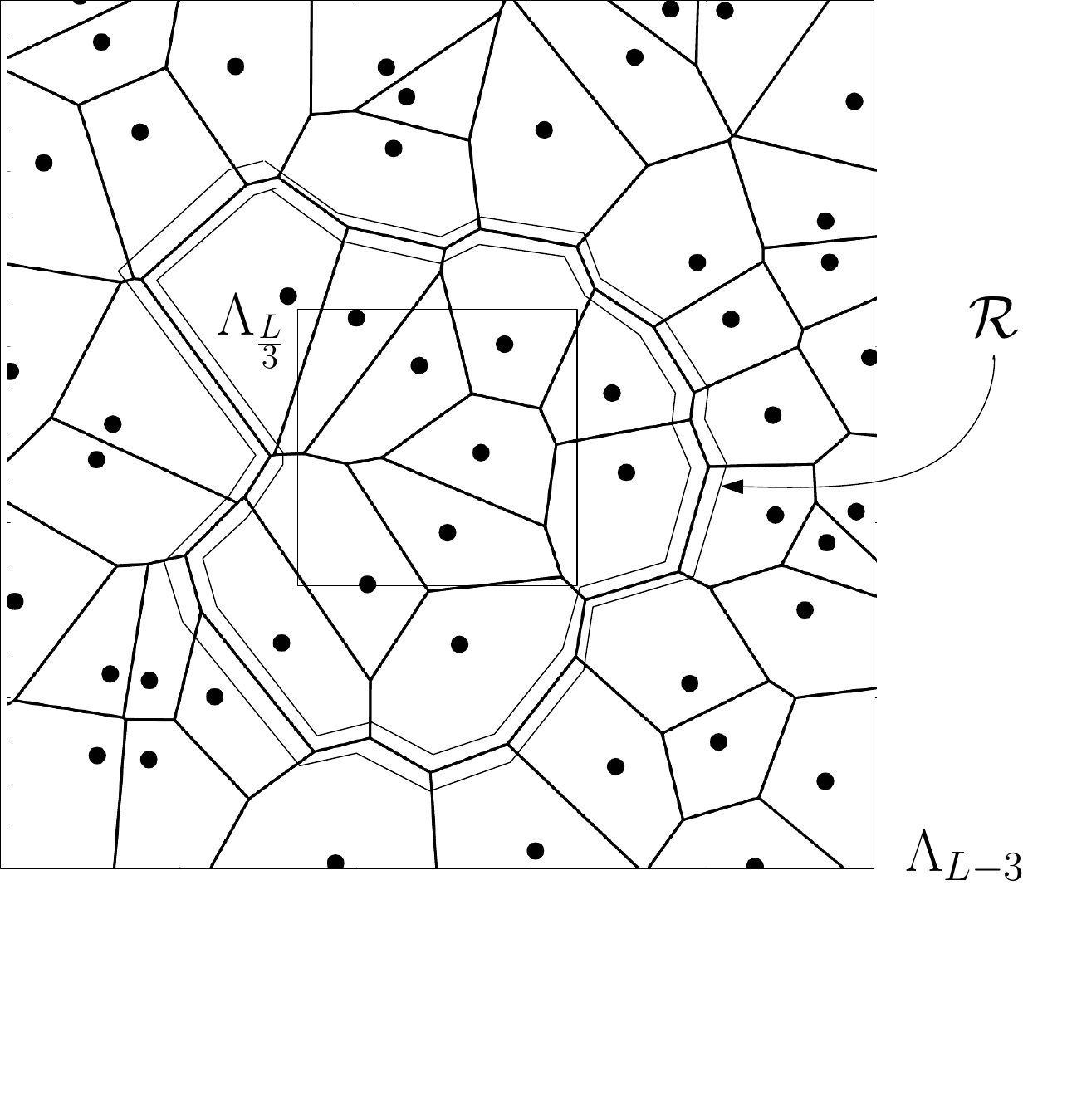}}
\caption{Ribbon $\mathcal R$ in the Voronoi diagram associated to $D$.  Points represent the support of the Delone-Anderson potential.}
\label{fig:Voronoi}
\end{figure}

Then, condition \cite[Eq. (4.2)]{CH} holds almost surely, therefore \cite[Corollary 4.1]{CH} holds almost surely, and this implies (see \cite[Proposition 5.1]{CH}, \cite[Theorem 4.3]{GK2})

\begin{thm}\label{ILE}
Let $E=B_n\pm 2a$ for some $n=0,1,2...$ with $0<2a<B$.  There exists
constants $Y_n,\beta_n>0$ depending only on $n,M,u, \delta_u$ such
that for any $0<\epsilon\leq a$, $L\in6\mathbb N$  and $Q_n$ as in
the previous theorem,

\be \mathbb P\left\{\Vert \Gamma_{x,L}
R_{B,\omega,x,L}(E)\chi_{x,L/3} \Vert\leq Y_n\frac{B}{a\epsilon^2}
e^{-\beta_n min\{aB,\sqrt{B}\}}\right\}> 1-Q_n\frac{B\epsilon
}{a^2}L^2.\ee

\end{thm}

Therefore, to satisfy (\ref{suitable}) we need only to verify
the conditions

\bea Y_n\frac{B}{a\epsilon^2} e^{-\beta_n
min\{aB,\sqrt{B}\}}\leq\frac{1}{L_0^\theta},\\
Q_n\frac{B\epsilon }{a^2}L_0^2\leq \frac{1}{L_0^p}, \eea
which can be done in the same way as in the proof of \cite[Theorem 4.1]{GK2}, yielding Theorem \ref{LocLandau} .

\bigskip


\subsection{Dynamical delocalization in Landau bands}

\begin{thm}\label{DelocLandau} Under the disjoint bands condition (\ref{DBC}) the random Landau
Hamiltonian $H_{B,\lambda,\omega}$ exhibits dynamical delocalization
in each Landau band $\mathcal B_n(B,\lambda)$, i.e. for all $n=1,2,...$,

\be \Xi^{DD}\cap \sigma_{B,\lambda,\omega}\cap \mathcal B_n(B,\lambda)\neq
\emptyset. \ee

In particular, there exists at least one energy $E_{n,\omega}(B,\lambda)\in
\mathcal B_n(B,\lambda)$ such that for every $\mathcal X \in
C^\infty_{c,+}(\mathbb R)$ with $\mathcal X \equiv 1$ on some open
interval $J\ni E_{n,\omega}(B,\lambda)$ and $p>0$, we have

\be \mathcal M_{B,\lambda}(p,\mathcal X,T)\geq C_{p,\mathcal
X}T^{\frac{p}{4}-6}, \ee
for all $T\geq 0$ with $C_{p,\mathcal X} >0$.

\end{thm}

This is a consequence of the quantization of the Hall conductance in each Landau band  and the fact that in regions of dynamical localization, the Hall conductance is constant, as proven in \cite[Section 3]{GKS}.  We recall the main lines of their strategy.

Consider the switch function $h(t)=\chi_{[\frac 1 2, \infty)}(t)$ and let $h_j$ denote the multiplication by the function $h(x_j)$, $j=1,2$.  The Hall conductance is defined as 

\be \sigma_{H_\omega}(B,\lambda,E)= -2\pi i \Theta(P_{B,\lambda,\omega,E}):= \tr \{ P_{B,\lambda,\omega,E}[[P_{B,\lambda,\omega,E},h_1],[P_{B,\lambda, \omega,E},h_2]]\} \ee
where $P_{B,\lambda,\omega, E}:=P_{B,\lambda,\omega}((-\infty, E])$.

Following the proof of \cite[Lemma 3.2]{GKS} we see that the Hall conductance is constant in connected components of the dynamical localization region, where property SUDEC is valid, as consequence of Theorem \ref{Bootstrap}.  On the other hand, it is well known that for $\lambda=0$, $\sigma_{H_\omega}(B,\lambda,E)=n$ if $E\in (B_n,B_{n+1})$ for all $n=0,1,2,...$.  Under the disjoint bands condition (\ref{DBC}),  if $E\in \mathcal G_n(B,\lambda_*)$ for $\lambda_*$ and some $n\in\{0, 1, 2,...\}$, we can find some $\lambda_E>\lambda_*$ such that $E\in \mathcal G_n(B,\lambda)$ for all $\lambda \in [0,\lambda_E]$.  That is, the spectral gaps stay open as $\lambda$ increases. Then we prove along the lines of \cite[Lemma 3.3]{GKS} that $\sigma_{H_\omega}(B,\lambda,E)=n$  if $E\in \mathcal G_n(B,\lambda)$,  for all $[0,\lambda_E]$. 
 As the spectral gaps $\mathcal G_n(B,\lambda)$ are by definition part of the localization region, this implies that the Hall conductance has the same value in different gaps, which is a contradiction.  Therefore, we must have $\Xi^{DD}\cap \sigma_{B,\lambda,\omega}\cap \mathcal B_n(B,\lambda)\neq \emptyset$ for every $\omega\in\Omega$.

By Theorems \ref{LocLandau} and \ref{DelocLandau} we conclude that there exists a dynamical transition energy in each Landau band as stated in Theorem \ref{DelocLandau}.

\bigskip

\subsection{Almost sure existence of spectrum near band edges}\label{spectrum}

Since we deal with a non ergodic random operator, previous results on the nature of the spectrum do not hold in this setting. In particular, we cannot use the characterization of the spectra as a union of spectra of periodic operators as in \cite{GKS}.  We need a more constructive approach and thus, to go back to the argument used in \cite{CH}.  We extend  \cite[Theorem 7.1]{CH} to a Delone-Anderson potential to make sure that, although the spectrum $\sigma_{B,\lambda,\omega}$ is random, there exists almost surely some part of $\sigma_{B,\lambda,\omega}$ in the region were we can prove dynamical localization, that is, in the spectral band edges.
 

We explicit the dependence on the $(r,R)$-Delone $D$ set by writing $V_\omega^D$ for the Delone-Anderson potential and $H_\omega^D$ for the corresponding random operator defined by (\ref{randop}).

 Consider the operator acting on $L^2(\R^2)$, 
$H_\omega^D=H_{B}+\lambda V_{\omega}^D$ where $\lambda>0$ and $V_{\omega}^D$ is defined as in \ref{ranpot1}.  Recall that
\be \label{ranpot1}
V_\omega^D(x)=\displaystyle\sum_{\gamma \in D}\omega_\gamma
u_\gamma,
\ee
where $D$ is an $(r,R)$-Delone set, the random variables $\omega_\gamma$ are i.i.d. with absolute continuous probability density $\mu$,  $\mbox{ supp }\mu=[-M,M]$ and $u_\gamma=u(x-\gamma)$. Assume moreover $u\in \mathcal C^2$,$\norm{u}_\infty=1$, $\mbox{supp }u\subset \L_r(0)$ and $u(0)=1$.  

\begin{thm}\label{existspect}
Under the disjoint bands conditions, for a random Landau Hamiltonian as stated before and any $n=0,1,2,...$ there exists a finite positive constant $B(n)$ depending on $n$, $M$, $u$, $ \lambda$ and $K_n(\lambda)$ such that for all $B>B(n)$, the intervals $\Sigma_{B,n,\lambda, \omega}$ in Theorem \ref{LocLandau} are almost surely non empty.
More precisely, we prove that there exist finite positive constants $C_n$, $B(n)$ depending on $n$, $M$, $u$ such that for every $B>B(n)$, we have for all $E\in\mathcal B_n$,
\begin{equation}
\sigma(H_{\omega})\cap [E-\lambda C_n B^{-1/2},E+\lambda C_n B^{-1/2}]\neq \emptyset
\end{equation}
\end{thm}

For a set $A\in\R^2$ we denote by $\tilde A$ the intersection $A \cap D$.  Recall that we have, for an arbitrary box $\L_L(x)$ of side $L\in\mathbb N$ centered in $x$:

\be \label{DeloneDense}  C_{R,d} L^d\leq \sharp(\tilde \L_L)=\sharp(D\cap \L_L)\leq C_{r,d} L^d, \ee
where $C_{R,d}=R^{-d}$ and $C_{r,d}= \lceil r^{-d} \rceil$.

Take a sequence $\{x_n \}$ such that $|x_n-x_m|>L$ for every $n, m$ and consider the following sets in the probability space $\Omega$:
\[\Omega^L_{\epsilon}(x_n)=\{\omega :|\omega_\gamma-\eta|\leq \epsilon \mbox{ } \forall \gamma\in \tilde\Lambda_L(x_n) \}\]
and 
\be \Omega^L_{\epsilon}=\bigcap_{N}\bigcup_{n\geq N}\Omega^L_{\epsilon}(x_n) \ee
where $\eta\in [-M,M]$.  By the choice of $\{x_n \}$, the events $\Omega_{\epsilon}^L(x_n)$ and $\Omega_{\epsilon}^L(x_m)$ are independent for  $n\neq m$.  

Since the random variables are i.i.d. and (\ref{DeloneDense}) holds for every box $\L_L(x_n)$, we obtain

\begin{align} \Prob{\Omega^L_{\epsilon}(x_n)}=\Prob{|\omega_\gamma-\eta|\leq \epsilon, \forall \gamma\in\tilde \Lambda_L(x_n)} & = \Prob{|\omega_\gamma-\eta|\leq \epsilon}^{\sharp(D\cap \Lambda_L(x_n))}\\
& \geq \Prob{|\omega_\gamma-\eta|\leq \epsilon}^{C_{r,d}L^d}\\
&= \mu([\eta-\epsilon,\eta+\epsilon])^{C_{r,d}L^d}\\
\end{align}
Therefore
\be \sum_{n} \Prob{\Omega^L_{\epsilon}(x_n)} =\infty,\ee
which implies that $\Prob{\Omega^L_{\epsilon}}=1$, by the Borel-Cantelli lemma.

Given $\delta>0$, take $\epsilon=\delta/(rL)^d$.  We have shown that for $\omega\in\Omega_\epsilon^L$, a set of full measure, there exists an infinite sequence $\{x_n\}$ such that for any $\eta\in [-M,M]$,

\be \label{BC}|\omega_\gamma - \eta|< \frac \delta {(rL)^d} \quad \mbox{for all }\gamma\in \tilde \L_L (x_n) \ee 

 Fix one of these boxes and call it $\L_0$ (so $\L_0$ depends on $\omega$, but this procedure can be done for all $\omega\in\Omega_0$, the yielding result being uniform in $\omega$). 

\smallskip

Without loss of generality, $\tilde \L_0$ contains $0$.  Indeed, if $0\notin \tilde \L_L(x_n)$ for all $n$, take $L>R$ so that $\tilde \L_0 \neq 0$ and take $\gamma_0 \in \tilde \L_0$.  Consider now the operator 
\be H_\omega^{D-\gamma_0}=H_{B}+\lambda \displaystyle\sum_{\gamma \in D-\gamma_0}\omega_\gamma
u_\gamma\ee 

We have that $\sigma (H_\omega^{D})=\sigma (H_\omega^{D-\gamma_0})$, since, taking a translation $\tau_{\gamma_0}: \Omega\times D\rightarrow \Omega\times (D-\gamma_0)$ defined by $\tau_{\gamma_0}(\omega_\gamma, \gamma)=(\omega_\gamma, \gamma-\gamma_0)$, that  associates the same random variable of a point to its translated, we can see $H_\omega^{D}$ is unitarily equivalent to $H_\omega^{D-\gamma_0}$.

Moreover,  by what is known for $H_\omega^D$, with full probability there exists a sequence $\{\tilde x_n\}=\{x_n-\gamma_0\}$ such that (\ref{BC}) holds.  In particular, since the cube $\L_0$ is a cube that satisfies (\ref{BC}) for $H_\omega^D$, then the cube  $\L_{\gamma_0}=\L_0-\gamma_0$ 
satisfies (\ref{BC}) for $H_\omega^{D-\gamma_0}$.

Define
 
\be \label{V0}
V_{\gamma_0}(x)=\eta\displaystyle\sum_{\gamma \in \tilde \L_{\gamma_0}}
u_\gamma.
\ee
Since $\gamma_0 \in \tilde \L_0=\L_0 \cap D$ we have that $0\in \tilde \L_{\gamma_0}=(\L_0-\gamma_0) \cap (D-\gamma_0)$.  Moreover, the assumptions on $u$, namely that $u(0)=1$ and the supports of $u_\gamma$ do not overlap, imply that $V_{\gamma_0}(0)=\eta$.  Therefore, without loss of generality we can assume $\tilde \L_0$ is centered in $0$ and so we work from now on with $H_\omega^D$, $V_\omega^D$ and $V_0$ as in (\ref{V0}) with $\gamma_0=0$.

\begin{rem}
\emph{The assumption $u(0)=1$ is so we can later perform a Taylor expansion around $0$}.
\end{rem}

\begin{proof}[Proof of Theorem \ref{thmCH}]
From now on $L$ is fixed.  For the sake of completeness, we will reproduce the details of \cite[Appendix 2]{CH} with the corresponding adaptations and work in the \emph{0}-th Landau band.  Let $\Pi_0$ be the Landau projection in the \emph{0}-th Landau band, around the Landau level $B_0$.  Take the normalized function $\phi_0\in \Pi_0(\mathcal H)$,defined by

\be \phi_0(x)=\left(\frac{2B}{\pi}\right)^{1/2}e^{-B|x|^2}. \ee

Let $E\in [B_0-\lambda M, B_0+\lambda M]$, that is,  $E=B_0 + \lambda \eta$ for some $\eta \in [-M, M]$.  The case $\eta =0$ is trivial by the previous Borel-Cantelli argument, as $\{B_n\}_{n\geq 0}\subset \sigma(H_\omega)$ almost surely. Since the argument is analog for $\eta<0$,  in the following we consider only $\eta \in (0, M]$, and write

\begin{align}
\norm{ \left(H_{\omega}^{D}-E\right)\phi_0} & =
\norm{ \left(H_{\omega}^{D}-B_0-\lambda \eta\right)\phi_0} \\
& \leq \norm{\Pi_0(\lambda V_\omega^{D}-\lambda \eta)\phi_0}+\lambda \norm{(1-\Pi_0) V_\omega^{D} \phi_0}
\end{align}
For simplicity we write $V_\omega$ instead of $V_\omega^{D}$.
The deterministic result \cite[Lemma A.1 ]{CH} implies that
\be \label{A} \lambda \norm{(1-\Pi_0) V_\omega \phi_0}\leq \lambda C_1 B^{-1/2},\ee 
where $C_1$ is a constant depending only on the single-site potential $u$.
We are left with

\begin{align}
\norm{\Pi_0(\lambda V_\omega-\lambda \eta)\phi_0} \leq &\lambda \norm{ (\displaystyle\sum_{\gamma \in \tilde \L_0}
\omega_\gamma u_\gamma + \displaystyle\sum_{\gamma \in D\setminus \tilde \L_0}
\omega_\gamma u_\gamma -\eta)\phi_0 }\\
\leq & \lambda \norm{ (\displaystyle\sum_{\gamma \in \tilde \L_0}
\omega_\gamma u_\gamma-\eta)\phi_0 }+ \lambda \norm{\displaystyle\sum_{\gamma \in D\setminus \tilde \L_0}
\omega_\gamma u_\gamma\phi_0}\\
\leq & \lambda \norm{ (\displaystyle\sum_{\gamma \in \tilde \L_0}
\omega_\gamma u_\gamma-\eta)\phi_0 }+ \lambda M\displaystyle\sum_{\gamma \in D\setminus \tilde \L_0}
\norm{u_\gamma\phi_0} \label{T}
\end{align}

Recall that 
\be \{ \gamma \in D:\quad \gamma \in D\setminus \tilde\L_0\}\subset\{ \gamma \in D : \quad |\gamma|>r \}.\ee

The second term in (\ref{T}) can be estimated as in \cite[Eq. 7.6]{CH}, where it is shown that

\be \norm{u_\gamma \phi_0}^2 = \int_{\mathbb R^2} \phi_0(x)^2 u(x-j)^2 dx \leq \norm{u}^2_\infty e^{-2B|j|^2+4B r |j|} \ee
which is summable for $\gamma$ such that $|\gamma|>r$,  yielding that for all $B>B_*$, for a constant $B_*$ big enough,

\be \label{B} \lambda M\displaystyle\sum_{\gamma \in D\setminus \tilde \L_0}
\norm{u_\gamma\phi_0} \leq \lambda C_2 B^{-1/2}
\ee
where the constant is uniform in $B$.

As for the first term in (\ref{T}), recalling the definition of $V_0$ from (\ref{V0}), we write

\begin{align} \lambda \norm{ (\displaystyle\sum_{\gamma \in \tilde \L_0}
\omega_\gamma u_\gamma-\eta)\phi_0 }= & \lambda \norm{ (\displaystyle\sum_{\gamma \in \tilde \L_0}
\omega_\gamma u_\gamma- V_0 + V_0  -\eta)\phi_0 }\\
\leq & \lambda \norm{ (\displaystyle\sum_{\gamma \in \tilde \L_0}
\omega_\gamma u_\gamma- \eta\displaystyle\sum_{\gamma \in \tilde \L_0}
u_\gamma) \phi_0} + \lambda \norm{(V_0 -\eta)\phi_0 }\\
\leq & \lambda \norm{ \displaystyle\sum_{\gamma \in \tilde \L_0}
(\omega_\gamma-\eta)u_\gamma\phi_0} +\lambda \norm{(V_0 -\eta)\phi_0 } \label{T2}
  \end{align}

By the choice of $\L_0$ the first term in \ref{T2} is 
\be \label{C} \lambda \norm{ \displaystyle\sum_{\gamma \in \tilde \L_0}
(\omega_\gamma-\eta )u_\gamma\phi_0} \leq \lambda \delta \ee 

As for the second term in \ref{T2},

\begin{align} \norm{(V_0-\eta)\phi_0}^2 & = \left( \frac 2 \pi \right) \int_{\R^2} |V_0(x)-\eta|^2e^{-2B\norm{x}^2}dx \\
& =\left( \frac 2 \pi \right) \int_{\R^2} |V_0(B^{-1/2}x)-\eta|^2e^{-2\norm{x}^2}dx  \end{align}

Now, since $V_0(0)=\eta$, we have 
\be |V_0(B^{-1/2}x)-\eta|=|V_0(B^{-1/2}x)-V_0(0)| \ee 
and we can perform a Taylor expansion around $0$ for $V_0$, obtaining, since $\mbox{supp }V_0\subset \L_0$

\be |V_0(B^{-1/2}x)-V_0(0)|\leq B^{-1/2}\norm{x}\norm{\nabla V_0}_\infty\leq B^{-1/2}L\norm{\nabla V_0}_\infty\ee

Notice that $\norm{\nabla V_0}_\infty\leq C_3$, for a constant $C_3$ depending only on $u$, uniformly with respect to $\eta\in[0,M]$.  Replacing this in the integral we obtain

\be \norm{(V_0-\eta)\phi_0}^2= \left( \frac{C_4}{\pi B} \right) \int e^{-2\norm{x}^2}dx  \ee 

So we obtain once more 
\be \label{D} \lambda \norm{(V_0-\eta)\phi_0}\leq \lambda C_5 B^{-1/2} \ee 

Finally, adding the estimates (\ref{A}),(\ref{B}),(\ref{C}) and (\ref{D}) yields that for all $B>B_*$,

\be
\norm{H_{\omega}^{D}-(B_0+\lambda \eta)}\leq \lambda C_5 B^{-1/2}+ \delta
\ee
where the bound is uniform in $B$, $\omega\in \Omega_0$ and in $\eta\in [0,M]$.  The same result holds in any Landau band for all $B$ large enough.  Therefore, with probability one and for any $E=B_n+\lambda \eta$, we have  \be \sigma(H_{\omega}^{D})\cap [E-\lambda C_5 B^{-1/2}- \delta,E+ \lambda C_5 B^{-1/2}+ \delta ]\neq 0 \ee
 
Since $\delta>0$ is arbitrary, 


\be \sigma(H_{\omega}^{D})\cap [E-\lambda C_5 B^{-1/2},E+ \lambda C_5 B^{-1/2}]\neq 0,\ee 
for every $E\in [B_n, B_n+\lambda M]$.  This proves that any gap in the spectrum of $H_\omega^D$ in the Landau band cannot exceed a length of order $B^{-1/2}$.
\end{proof}

In particular, since we know by perturbation theory that $\sigma(H_{\omega}^{D}) \subset [B_n-\lambda M, B_n+\lambda M]$,we have that for $E=B_n+\lambda M$, that is, in the edge of the Landau band, 

\be \label{S}\sigma(H_{\omega}^{D})\cap [B_n+\lambda M-\lambda C_5 B^{-1/2},B_n+\lambda M]\neq \emptyset\ee

On the other hand, by Theorem \ref{LocLandau} we know the localization region is at a distance $K_n(\lambda)\frac {\ln B} B$ from the Landau level $B_n$.  If $\lambda$ is fixed and $B$ is such that 

\be K_n(\lambda) \frac {\ln B} B < \lambda M - \frac{\lambda C_n}{\sqrt B},\ee 
then the region of the spectrum that is almost surely near the band edge, that is above $B_n+\lambda M-\lambda C_n B^{-1/2}$, lies in the localization region, that is above $B_n +K_n(\lambda) \frac {\ln B} B $.  So we have shown Theorem \ref{existspect}, that is, for every $n=0,1,2,...$

\be \Sigma_{B,n, \lambda,\omega} \neq \emptyset \quad \mbox{for a.e. }\omega\in\Omega \ee


\end{document}